\documentclass[twoside,10pt]{amsart}
\usepackage{graphicx,amsmath,amsfonts,amsthm,amssymb, fancyhdr,verbatim}
\newtheorem{thm}{Theorem}[section]
\newtheorem{cor}[thm]{Corollary}
\newtheorem{lem}[thm]{Lemma}
\newtheorem{prop}[thm]{Proposition}

\theoremstyle{definition}

\theoremstyle{remark}

\def\beq{\begin{eqnarray}}
\def\eeq{\end{eqnarray}}
\def\bsp{\begin{split}}
\def\esp{\end{split}}

\newcommand{\z}{\zeta}
\newcommand{\bz}{\bar{\zeta}}

\newcommand{\ba}{\bar{a}}

\begin{document}

\title{\Large\textbf{Vacuum Kundt waves }}
\author{{\large\textbf{David McNutt$^{1}$, Robert Milson$^{1}$, and Alan Coley$^{1}$ } }
 \vspace{0.3cm} \\
$^{1}$Department of Mathematics and Statistics,\\
Dalhousie University,
Halifax, Nova Scotia,\\
Canada B3H 3J5
\vspace{0.2cm}\\
\texttt{ddmcnutt@dal.ca, aac@mathstat.dal.ca,rmilson@dal.ca } }
\date{\today}
\maketitle
\pagestyle{fancy}
\fancyhead{} 
\fancyhead[EC]{D. McNutt, R. Milson and A. Coley}
\fancyhead[EL,OR]{\thepage}
\fancyhead[OC]{Vacuum Kundt Waves}
\fancyfoot{} 

\begin{abstract} 

We discuss the invariant classification of vacuum Kundt waves using the Cartan-Karlhede algorithm and determine the  upper bound on the number of iterations of the Karlhede algorithm to classify the vacuum Kundt waves \cite{Collins91, MRVickers}. By choosing a particular coordinate system we partially construct the canonical coframe used in the classification to study the functional dependence of the invariants arising at each iteration of the algorithm. We provide a new upper bound, $q \leq 4$, and show that this bound is sharp by analyzing the subclass of Kundt waves with invariant count beginning with (0,1,...) to show that the class with invariant count $(0,1,3,4,4)$ exists. This class of vacuum Kundt waves is shown to be unique as the only set of metrics requiring the fourth covariant derivatives of the curvature. We conclude with an invariant classification of the vacuum Kundt waves using a suite of invariants. 
\end{abstract} 

\maketitle


\begin{section}{Introduction}
The Kundt waves were originally defined by Kundt in 1961 \cite{Kundt61}, as a special subcase of the class of pure radiation  solutions of Petrov type III or higher and Plebanski-Petrov (PP) type O or vacuum admitting a non-twisting, non-expanding null congruence, $\ell$, that is 
\beq \ell^a\ell_a = 0,~~\ell^a_{~;a} = 0,~~\ell_{(a;b)}\ell^{a;b} =0,~~\ell_{[a;b]}\ell^{a;b} = 0. \nonumber \eeq 
\noindent These conditions restrict the Petrov type for the plane-fronted waves to Petrov type N or O. 
Choosing Kundt coordinates, the metric for the Kundt waves is 
\beq & ds^2 = d\z d\bz - du \left(dv - \frac{2v}{\z + \bz} (d\z + d\bz) + \left(4H(\z,\bz,u) (\z+\bz)- \frac{v^2}{(\z+\bz)^2} \right) du  \right), & \label{RPFWmetric} \eeq
\noindent where $u,v$ are null coordinates, and $\z, \bz$ are complex coordinates for the transverse space\cite{ExactSolns}. 

All polynomial curvature invariants, built from contracting the Riemann tensor and covariant derivatives with each other, vanish for these spacetimes. Thus, the plane-fronted belong to the collection of $VSI$ spacetimes where all polynomial curvature invariants vanish \cite{4DVSI}; this is, in turn, a subclass of the $CSI$ spacetimes in which all polynomial curvature invariants are constant \cite{4DCSI}. These spaces have been explored in four dimensions and were shown to belong to the class of degenerate Kundt metrics \cite{Kundt}. These are the Kundt metrics where the frame used to classify the Riemann tensor (i.e., Petrov or Riemann type \cite{Weyl}) and the kinematic frame are aligned; i.e., they are the same. It is expected that this is the case in higher dimensions as well \cite{HDVSI, Kundt}.

For a given spacetime in four dimensions, a spacetime is either uniquely determined by its polynomial scalar curvature invariants, is a (locally) homogeneous space, or is a degenerate Kundt spacetime \cite{Kundt}. For the degenerate Kundt  spacetimes the equivalence problem is particularly relevant,  given that one cannot determine the inequivalence of two metrics of this class by comparing polynomial scalar curvature invariants \cite{4DVSI, 4DCSI,4DKundt}. To invariantly classify these spacetimes, one must use an alternative tool, the Karlhede algorithm, which utilizes the Cartan equivalence method \cite{Cartan} adapted to the case of Lorentzian manifolds \cite{Karlhede} . 

The first and second stages of the Karlhede algorithm were analyzed for all type N vacuum spacetimes with $\Lambda =0$ by Collins \cite{Collins91}, who produced a theoretical upper bound on the highest order, q, of the covariant derivatives of the curvature tensor required for each of the various subclasses of the type N spacetimes. Interestingly, this gives a hard upper bound for the $VSI$ spacetimes \cite{4DVSI,4DCSI}, as the pp-waves and vacuum Kundt waves make up the entirety of type N $VSI$ spacetimes \cite{Ozvath, 4DVSI, Bicak}. Collins has shown that the pp-waves require $q \leq 4$ while the vacuum Kundt waves need at most $q \leq 6$. Recently it has been shown that the pp-wave upper bound is sharp \cite{Milson}, and that the actual Kundt-wave's upper bound is five \cite{MRVickers}. However, in 2000, Skea produced a {\it non-vacuum} Kundt wave in which $q=5$, suggesting that there might be vacuum solutions for which $q=5$ \cite{skea}.

In this paper, we discuss the upper bound for the vacuum Kundt waves in the Karlhede algorithm or, equivalently,  the highest order, q, covariant derivative of the curvature required to invariantly classify these spaces. We show that the upper bound may be lowered to be less than or equal to four by exploring all possible outcomes of the Karlhede algorithm (see figures  \eqref{0KarlAlg}, \eqref{2KarlAlg} and \eqref{1KarlAlg}). Out of all possible invariant counts only one actual vacuum Kundt wave may be integrated; namely, the class with invariant count $(0,1,3,4,4)$. Due to the exhaustive nature of this analysis we examine the remaining branches of possibilities in the algorithm to produce an invariant classification of all vacuum Kundt waves. This classification is summarized in two tables describing each of the non-diffeomorphic vacuum Kundt wave metrics arising by the choice of the metric function $f(\z,u)$. We present twelve propositions relating the form of the metric function $f(\z,u)$ to the essential Cartan invariants characterizing each spacetime in the first three appendices. The final appendix contains all of the potential subcases of the Karlhede algorithm applied to the vacuum Kundt wave spacetimes prior to examining the geometric structure of these spacetimes.
\end{section}

\begin{section}{Geometric Structure of the Vacuum Kundt Waves}
If we wish to preserve the form of the metric, the permitted coordinate transformations are \cite{4DVSI}:
\beq & \z' = \z +i\tilde{C},~~u'=h(u), v' = \frac{v}{h_{,u}}-(\z + \bz)^2 \frac{h_{,uu}}{2h_{,u}^2},& \label{KWcoordtransf} \\
& H' = \frac{H}{h_{,u}^2}+\frac{(\z+\bz)}{4h_{,u}^4}(-3h^2_{,uu}+2h_{,u}h_{,uuu}),& \nonumber \eeq
\noindent where $\tilde{C}$ is a real constant and $h(u)$ is an arbitrary real function. Taking the metric \eqref{RPFWmetric}, we work with the Newman-Penrose formalism \cite{NP} to calculate the non-vanishing curvature components of the Ricci ($\Phi$) and Weyl $(\Psi$) spinors, respectively: 
\beq \Phi_{22} = x H_{,\z \bz};~~\Psi_4 = 2H_{,\bz \bz}. \nonumber \eeq
\noindent To satisfy the vacuum conditions, $H$ must be harmonic and real-valued; as in the pp-waves, this will be the real part of an analytic function, $2H = f(\z,u) + \bar{f}(\bz,u)$. To examine the geometric structure of these spaces, we work with the class of coframes in which $\Psi_4 =1$. These are found by applying an appropriate spin and boost to the natural metric coframe. 

Without imposing the vacuum condition, the non-vanishing Bianchi identities imply the relationship between the spin-coefficients and the components of the Ricci and Weyl spinors \cite{NP} and their frame derivatives $D,\Delta,\delta,\bar{\delta}$:
\beq &\kappa = \sigma = \rho = 4 \epsilon = 0,~D\Phi_{22} = 0, & \nonumber \\
 & \bar{\delta} \Phi_{22}  = (4\beta - \tau)\Psi_4 + (\bar{\tau} - 2 \bar{\beta}-2\alpha)\Phi_{22}. & \nonumber \eeq
\noindent Imposing the vacuum conditions, we see that $\beta = \frac{\tau}{4}$.  
The non-vanishing Newman-Penrose field equations for the vacuum Kundt waves  are: 
\beq & D\tau =0,~~D\alpha = 0& \label{VacNPbc} \\
&  D\gamma = \frac54 \tau \pi + \tau \alpha + \bar{\pi} \alpha + \frac14 \tau \bar{\tau}, & \label{VacNPd} \\
& D \lambda - \bar{\delta} \pi = \pi^2 + \alpha \pi - \frac14 \bar{\tau} \pi, & \label{VacNPe} \\ 
& D \mu - \delta \pi = \pi \bar{\pi} - \pi \bar{\alpha} + \frac14 \pi \tau, & \label{VacNPf} \\
& D \nu - \Delta \pi = \pi \mu + \bar{\tau} \mu + \bar{\pi} \lambda + \tau \lambda + \gamma \pi - \bar{\gamma}\pi, & \label{VacNPg} \\
& \Delta \lambda - \bar{\delta} \nu = - \mu \lambda - \bar{\mu} \lambda - 3\gamma \lambda + \bar{\gamma} \lambda + 3 \alpha \nu + \pi \nu - \frac34 \bar{\tau} \nu - \Psi_4, & \label{VacNPh} \\
& \delta \alpha - \frac14 \bar{\delta} \tau = \alpha \bar{\alpha} + \frac{1}{16} \tau \bar{\tau} - \frac12 \alpha \tau, & \label{VacNPj} \\
&\delta \lambda - \bar{\delta} \mu = \mu \pi - \bar{\mu} \pi + \mu \alpha + \frac14 \mu \bar{\tau} + \lambda \bar{\alpha} - \frac34 \lambda \tau, & \label{VacNPk} \\
& \delta \nu - \Delta \mu = \mu^2 + \lambda \bar{\lambda} + \gamma \mu + \bar{\gamma} \mu - \bar{\nu} \pi + \frac14 \tau \nu - \bar{\alpha} \nu, &  \label{VacNPl} \\
& \delta \gamma - \frac14 \Delta \tau = \frac12 \tau \gamma - \bar{\alpha} \gamma + \frac54 \mu \tau + \frac14 \tau \bar{\gamma} + \alpha \bar{\lambda}, & \label{VacNPm} \\
& \delta \tau = \frac54 \tau^2 - \tau \bar{\alpha}, & \label{VacNPn} \\
& - \bar{\delta} \tau = -\frac34 \bar{\tau} \tau - \alpha \tau,& \label{VacNPo} \\
& \Delta \alpha - \bar{\delta} \gamma = - \frac54 \tau \lambda + \bar{\gamma} \bar{\alpha} - \bar{\mu} \bar{\alpha} - \frac34 \tau \gamma, \label{VacNPp} \eeq
\noindent while the commutator relations are 
\beq  (\Delta D - D \Delta) f &=& [(\gamma + \bar{\gamma}) D - (\tau + \bar{\pi})\bar{\delta} - (\bar{\tau}+\pi)\delta] f, \nonumber \\
(\delta D - D \delta) f &=& [(\bar{\alpha} + \frac{\tau}{4} - \bar{\pi})D]f,  \nonumber \\
 \delta \Delta - \Delta \delta) f &=& [-\bar{\nu} D + (\frac{3\tau}{4} - \bar{\alpha}) \Delta + \bar{\lambda} \bar{\delta} + (\mu - \gamma + \bar{\gamma})\delta ]f,  \nonumber \\
 (\bar{\delta} \delta - \delta \bar{\delta}) f &=& [(\bar{\mu} + \mu) D - (\bar{\alpha} - \frac{\tau}{4})\bar{\delta} - (\frac{\bar{\tau}}{4} - \alpha) \delta] f. \nonumber \eeq

The benefit of working in the class of coframes for which $\Psi_4 =1$ becomes apparent once one takes frame derivatives of the Weyl tensor, as only spin-coefficients and their derivatives appear as components of the Weyl tensor and its covariant derivatives. To illustrate, the first order derivatives of the Weyl tensor are
\beq & (D \Psi)_{50'} = 4 \alpha,~~ (D \Psi)_{51'} = 4 \gamma,~~ (D \Psi)_{40'} = 0 ,& \nonumber \\
& (D \Psi)_{41'} = \tau ,~~ (D \Psi)_{30'} = 0, ~~ (D \Psi)_{31'} = 0. & \nonumber \eeq

\noindent At first order, one still has 2 degrees of frame freedom, using null rotations with complex parameter $B$, which affects the first order invariant $\gamma$ and leaves $\alpha$ and $\tau$ unchanged: 
\beq \gamma' = \gamma + B \alpha + \frac54 \bar{B} \tau. \label{NullRot} \eeq
\noindent If $|\alpha| \neq \frac54 |\tau|$ it is always possible to set $\gamma' = 0$. However, if equality holds, only one degree of freedom can be fixed, and there are three subcases for the form of $\gamma'$ \cite{Collins91}: 
\begin{itemize} 
 \item $\bar{\alpha} = -\frac54 \tau$: $Im(\gamma') = 0$;
 \item $\bar{\alpha} = \frac54 \tau$: $Re(\gamma') = 0$;
 \item $\bar{\alpha} \neq \pm \frac54 \tau$: $Re(\gamma')$ or $Im(\gamma')$ = 0, but not both.  
\end{itemize}
\noindent Without fixing the frame freedom, the non-zero second order derivatives of the Weyl tensor are: 
\beq (D^2 \Psi)_{50';00'} & = &4D \alpha, \nonumber \\
 (D^2 \Psi )_{50';01'} & = &4( \delta \alpha + 5 \beta \alpha -  \bar{\alpha} \alpha) , \nonumber \\
 (D^2 \Psi )_{50';10'} & = &4( \bar{\delta} \alpha  + 5 \alpha^2), \nonumber \\
 (D^2 \Psi)_{50';11'} & = &4( \Delta \alpha + 5 \gamma \alpha - \bar{\gamma} \alpha +  \bar{\tau} \gamma),   \nonumber \\
 (D^2 \Psi)_{51';00'}  & =& 4 (D \gamma - 5 \pi \beta - \bar{\pi} \alpha),   \nonumber \\
 (D^2 \Psi)_{51';01'} & = &4( \delta \gamma - 5 \mu \beta + 5 \beta \gamma - \bar{\lambda} \alpha +  \bar{\alpha} \gamma),    \nonumber \\
 (D^2 \Psi)_{51';10'} & = &4( \bar{\delta} \gamma - 5 \lambda \beta + 5 \alpha \gamma - \bar{\mu} \alpha +  \bar{\beta} \gamma),   \nonumber \\ 
 (D^2 \Psi)_{51';11'} & = &4( \Delta \gamma - 5 \nu \beta + 5 \gamma^2 - \bar{\nu} \alpha +  \bar{\gamma} \gamma),    \nonumber \\
 (D^2 \Psi)_{40';11'} & = & 4( \tau \alpha + \bar{\tau} \beta),   \nonumber \\
 (D^2 \Psi)_{41';00'} & = &4 D \beta ,  \nonumber \\
 (D^2 \Psi)_{41';01'} & = &4( \delta \beta + 3 \beta^2 + \bar{\alpha} \beta), \nonumber \\
 (D^2 \Psi)_{41';10'} & = &4(\bar{\delta} \beta + 3 \alpha \beta + \bar{ \beta}\beta), \nonumber \\
 (D^2 \Psi)_{41';11'} & = &4(\Delta \beta + 3 \gamma \beta + \tau \gamma  + \bar{\gamma} \beta),  \nonumber \\
 (D^2 \Psi)_{31';11'} & = &8 \tau \beta. \nonumber \eeq

\noindent If $|\alpha| = \frac54 |\tau|$, it is always possible to fix the last parameter of the frame freedom to fix $\Delta \tau$ so that $Re(\Delta \tau) = 0$. Manipulating the spin-coefficients and the remaining degrees of freedom, Collins produced a theoretical upper bound for these spaces \cite{Collins91}, requiring at most six covariant derivatives. This bound was lowered to five covariant derivatives by Machados Ramos and Vickers \cite{MRVickers} using the generalized GHP formalism. In both papers a particular choice of coordinates was avoided so that these bounds were not shown to be sharp.  
\end{section}

\begin{section}{An Alternative Proof That The Upper Bound for the Karlhede Algorithm is Less than Six} \label{Leq6}

The Karlhede algorithm terminates if and only if the dimension of the isotropy group and 
number of functionally independent invariants are unchanged from one iteration to 
the next. Using the invariant count notation, it is possible to map out all possibilities 
for the Karlhede algorithm. The case where the invariant count begins with $(0,0,...)$ is not permitted as the invariant $\tau$ must be non-constant at first order;  if we assume $\tau$ is a constant we find from \eqref{VacNPn} and \eqref{VacNPo} that $\tau =0$, which cannot be true since we are studying the vacuum Kundt waves. With this in mind, it is easily shown that there is only one scenario where $q=6$ at most, (0,1,1,2,3,4,4)\footnote{This notation is adopted in Appendix D to summarize possible states of the Karlhede algorithm compactly.}.

This invariant count would occur for the class of vacuum Kundt waves in which at first order only one functionally independent invariant appears and further that $|\alpha| = \frac{5|\tau|}{4}$. By choosing a particular coordinate system we may produce differential constraints on the metric function $H(\z,\bz,u) = Re(f(\z,u))$ by imposing the vanishing of the wedge products of the differentials of the spin-coefficients of $\alpha$, $\tau$ and their conjugates. As the spins and boosts have been fixed to set $\Psi_4 =1$, and since these two invariants $\alpha$ and $\tau$ are unchanged under the remainder of the isotropy group (null rotations about $\ell$), these are already Cartan invariants. With a little effort and a change of coordinates we intend to prove the following theorem:

\begin{thm} \label{thm:MRVickerBound}
The vacuum Kundt waves require at most $q=5$ iterations of the Karlhede algorithm to completely classify the spacetimes. 
\end{thm}

To this end we introduce a new complex coordinate $a = \frac14 ln( f_{,\z \z} )$. Relative to this new coordinate system, $\z = \z(a,u)$ and we find a differential constraint for the metric function $\tilde{f}$, \beq \left( \frac{\tilde{f}_{,a}}{\z_{,a}}\right)_{,a} = e^{4a}\z_{,a}. \label{aKundt} \eeq 
\noindent The metric coframe becomes: 
\beq m &=& \z_{,a } da + \z_{,u} du, \nonumber \\
\ell &=& du, \label{Mecoframe} \\
n &=& dv - \frac{2v}{\z+\bz} ( \z_{,a} da + \z_{,u} du) - \frac{2v}{\z+\bz} \bz_{,\ba} d\ba + \bz_{,u} du) \nonumber \\
& &+ \left(2 Re(\tilde{f}(a,u))(\z+\bz) - \frac{v^2}{(\z+\bz)^2} \right)du. \nonumber \eeq
\noindent In these coordinates the non-zero component of the Weyl tensor is now 
\beq \Psi_4 = 2(\z+\bz)e^{4b}. \nonumber \eeq
\noindent Applying a spin and boost with $p= \frac14 ln(|\bar{\Psi}_4|) = a + \frac14ln(2(\z+\bz))$ to the metric coframe \eqref{Mecoframe}, we produce a new coframe:
\beq m' = e^{p-\bar{p}} m,~~ \ell' = e^{p+\bar{p}} \ell,~~n' = e^{-p-\bar{p}} n. \label{Kacoframe} \eeq
\noindent Relative to this coframe, the non-vanishing Weyl tensor component has been normalized $\Psi'_4 =1$ and the spin-coefficients $\alpha$ and $\tau$ are already Cartan invariants as they are unaffected by the remaining isotropy. 

By direct calculation we produce the following spin-coefficients relative to this coframe:
\begin{prop} \label{prop:SpinCoef1}
The spin-coefficients relative to the class of coframes \eqref{Kacoframe}, in which $\Psi_4 = 1$, may be expressed as 
\beq & \tau = 4\beta = -\bar{\pi} = -\frac{e^{\ba -a}}{\z+\bz}, & \nonumber \\ 
 & \mu = \lambda = 0,& \nonumber \\
& \alpha = \frac{\bar{\tau}}{4}+ \sqrt{\frac{\bar{\tau}}{\tau}} (\bz_{,\ba})^{-1}, & \label{SpinCoef1} \\
&\gamma = -\frac{e^{-a-\ba} |\tau|^{\frac52}}{\sqrt{2}}\left( v + \frac{\bz_{,u}(\bz_{,a})^{-1}}{ |\tau|^{2}} \right), & \nonumber \\
&\nu = e^{-a-3\ba} \left( \int \bz_{,\ba} e^{4\ba} d\ba + f_1 - (f+\bar{f})|\tau|\right).& \nonumber \eeq
\end{prop}

\noindent Before we fix any more frame freedom to set all or a part of $\gamma$ to zero, we may determine the explicit form of the metric function $f(\z,u)$ for the class of vacuum Kundt waves where {\it only one} functionally independent invariant arises in the set $\{ \alpha, \tau, \bar{\alpha}, \bar{\tau} \}$:
\begin{lem} \label{lem:Aclass}
Those spacetimes in which the spin-coefficients $\alpha, \bar{\alpha}$, $\tau$ and $\bar{\tau}$ are functionally dependent on one invariant will have the following form for the metric function $f(\z,u)$:
\beq & f(\z,u) = \frac{C_0^2}{16} e^{\frac{-4i(\z-iG(u)+C_1)}{C_0}} + f_1(u) \z + f_2(u). & \label{AclassMetric} \eeq
\noindent Relative to the coordinates $a=\frac14 ln(f_{,\z \z})$, the Cartan invariants $\alpha$ and $\tau$ are now 
\beq & \tau = \frac{-e^{\ba-a}}{iC_0(a-\ba)+2C_1},~~\alpha = \frac{\bar{\tau}}{4} + \frac{i}{C_0}\sqrt{\frac{\bar{\tau}}{\tau}}. & \label{AclassTnA} \eeq
\end{lem}

\begin{proof}
Taking $\tau$ in \eqref{SpinCoef1}, we calculate the double wedge product of $d\tau$ and $d\bar{\tau}$ to get,
\beq & d \tau \wedge d \bar{\tau} = \frac{2}{(\z+\bz)^3}\left[(\z_{,a} + \bz_{,\ba}) da \wedge d\ba + (\z_{,u} + \bz_{,u}) da \wedge du + (\z_{,u} + \bz_{,u}) d\ba \wedge du \right]. & \nonumber \eeq
\noindent Requiring that this must vanish gives a set of equations: $\z_{,a} = -\bz_{,\ba},~~ \z_{,u} = -\bz_{,u}$.
\noindent Thus $\z(a,u)$ is of the form 
\beq & \z(a,u) = i(C_0a + G(u))+C_1. & \label{AclassZ} \eeq
\noindent Plugging this into the expressions for $\tau$ and $\alpha$ in \eqref{SpinCoef1} we recover \eqref{AclassTnA}, and then solving for $a$ and noting that $e^{4a} = f_{,\z \z}$ we may integrate twice to recover the function in the usual coordinate system. 
\end{proof}

The vacuum Kundt wave spacetimes with this property will potentially contain at most two functionally independent invariants at first order: $\tau$ and $\gamma$ which will simplify the search for those vacuum Kundt waves with only one functionally independent invariant at first order. Furthermore, as the necessary conditions for fixing the remaining isotropy is dependent on the Cartan invariants $\alpha$ and $\tau$ we may use the explicit form of these invariants from lemma \ref{lem:Aclass} to show all isotropy may be fixed at first order  (i.e., $|\alpha| \neq \frac54 |\tau|$) and that no vacuum Kundt wave requires $q=6$ in the algorithm.

\begin{cor} \label{upper bound5} 
The invariant count $(0,1,1,2,3,4,4)$ cannot occur in the Karlhede classification of the vacuum Kundt waves. 
\end{cor}

\begin{proof}
From lemma \ref{lem:Aclass} we calculate the equality $|\alpha| = \frac54 |\tau|$ using equation \eqref{AclassTnA}. We assume the equality holds and multiply both sides by $|\alpha|$, so that $|\alpha|^2 = \frac{25}{16} |\tau|^2$. Expanding this we have:
\beq \frac{25}{16} |\tau|^2 = \frac{1}{16} |\tau|^2 + \frac{1}{C_0^2}. \nonumber \eeq
\noindent Using the $a, \ba$ coordinates and simplifying we find the following
\beq \frac32 C_0^2 = (iC_0(a-\ba)+2C_1)^2. \nonumber \eeq
\noindent Differentiating with respect to $a$ or $\ba$ implies that $C_0 = 0$ which cannot happen as $\z_{,a}= C_0$ must be non-zero. This is a contradiction and so $|\alpha| \neq \frac54|\tau|$ .
\end{proof}

As this is the only permitted state in the Karlhede algorithm for the vacuum Kundt waves with $q=6$, and this case cannot occur, we conclude that the upper-bound for the vacuum Kundt waves may be lowered to less than or equal to five. 
\end{section}
\begin{section}{Reducing the Upper Bound to Less than Five} \label{Leq5}

The goal of this section is to provide the necessary lemmas to prove the following theorem: 

\begin{thm} \label{thm:4itis}
The vacuum Kundt wave spacetimes require, at most, four derivatives (i.e., $q=4$) to classify these spaces using the Karlhede algorithm.
\end{thm}

\noindent To study the sharpness of the upper bound, we examine the possible iteration scheme for the Karlhede algorithm applied to the vacuum Kundt waves as tree diagrams. This may be done exhaustively for the cases where there are at least one invariant at the first iteration of the algorithm. 
\begin{lem} \label{lem:CaseCounting}
The vacuum Kundt wave spacetimes for which the Karlhede algorithm requires five iterations have invariant counts \beq & (0,1,2,3,4,4),~~~and~~~ (0,2,2,3,4,4) &. \nonumber \eeq
\end{lem}
\begin{proof}
 The trees for the various possibilities are included in Appendix D.
\end{proof}
\noindent To prove theorem \ref{thm:4itis} we must examine the constraints on the vacuum Kundt waves to produce the invariant counts in lemma \ref{lem:CaseCounting}. To do so we break up the analysis into two subsections to examine the distinct subclasses of vacuum Kundt waves with either one or two functionally independent invariants appearing at first order. 
\begin{subsection}{Vacuum Kundt waves with $(0,1,2,..)$} \label{Aclassq4upb}
 Applying the results of lemma \ref{lem:Aclass} and corollary \ref{upper bound5}, we are able to say something about the upper bound in the first case, as the invariant coframe is produced from \eqref{Kacoframe} by making a null rotation \eqref{NullRot} to set $\gamma'=0$. We must determine the form of the parameter $B$ for the null rotation taking the coframe \eqref{Kacoframe} to the invariant coframe required for the Karlhede algorithm:
\beq & \ell' = \ell,~~n' = n+\bar{B} m+B \bar{m} + |B|^2 \ell,~~m'=m+B\ell. & \label{Icoframe} \eeq
\noindent To achieve this, we equate \eqref{NullRot} to zero and solve for $B$, 
\beq & B = - \sqrt{2} |\tau|^{\frac52} e^{-a-\ba} \sqrt{\frac{\tau}{\bar{\tau}}} \left( \frac{C_0^2 |\tau| - iC_0}{3C_0^2|\tau|^2 - 2} \right) \left(v+\frac{G_{,u}}{C_0|\tau|^2} \right). & \label{NRotB} \eeq

Using the dual of the invariant coframe, $\{ \delta', \bar{\delta}', \Delta', D'\}$,  we may compute the second order Cartan invariants as the frame derivatives of the first order Cartan invariants along with the following transformed spin-coefficients:
\beq \begin{aligned} \pi' &= \pi +D\bar{B},  \\ 
\lambda' &= \frac{\bar{B} \bar{\tau}}{2} + \sqrt{\frac{\bar{\tau}}{\tau}} \frac{2\bar{B}}{\bar{\z}_{,\ba}} + \bar{B} \pi+ \bar{B} D\bar{B} + \bar{\delta} \bar{B}, \\ 
\mu' &= \frac{\bar{B} \tau}{2} + B \pi +  B D\bar{B} + \delta \bar{B}, \\
\nu' &= \nu + 2\bar{B} \gamma + \frac32 \bar{B}^2 \tau + B\bar{B}(\pi + 2\alpha) + \Delta \bar{B} + \bar{B} \delta \bar{B} + B \bar{\delta} B + B \bar{B} D\bar{B}. \label{InvSpinC} \end{aligned} \eeq 
\noindent These remaining invariants are expressed in terms of the coframe \eqref{Kacoframe}, the original spin-coefficients \eqref{SpinCoef1}, and the frame derivatives of $B$ relative to the original coframe \eqref{Kacoframe} with $\Psi_4 =1$:
\beq D &=& \sqrt{\frac{2}{|\tau|}} e^{a+\ba} \frac{\partial}{\partial_v}, \nonumber \\
\Delta &=&  \sqrt{\frac{|\tau|}{2}}e^{-a-\ba}\left( \frac{\partial}{\partial_u} - \left( \frac{2(f+\bar{f})}{|\tau|}-v^2|\tau|^{2} \right) \frac{\partial}{\partial_v} - \frac{\z_{,u}}{\z_{,a}} \frac{\partial}{\partial_a} - \frac{\bz_{,u}}{\bz_{,\ba}} \frac{\partial}{\partial_{\ba}} \right), \label{Kaframe} \\
\delta &=& \frac{e^{a-\ba}}{\bz_{,\ba}} \frac{\partial}{\partial_{\ba}} - 2v \bar{\tau} \frac{\partial}{\partial_v}. \nonumber \eeq

Noting that $\bar{\pi} = -\tau$ in \eqref{SpinCoef1} and subtracting $-\tau$ from $\bar{\pi}'$, it is clear that $DB$ is an invariant; a quick calculation confirms that it is functionally dependent on $\tau$ and its conjugate.
\noindent We now examine the second order invariant arising from the frame derivative of $|\tau|^{-1}$;  
\noindent removing all terms that are functionally dependent on $\tau$ leaves the helpful invariant: 
\beq \xi = e^{-a-\ba}\left( v + \frac{G_{,u}}{C_0|\tau|^2} \right) \label{xi}. \eeq

As $|\alpha| \neq \frac54 |\tau|$, the remaining invariants at second order may be simplified to the spin-coefficients  $\mu', \lambda'$ and $\nu'$. These spin-coefficients involve $B$ and the remaining frame derivatives of this function:
\beq \begin{aligned} &  \bar{\delta} B =  \tau \sqrt{\frac{|\tau|}{2}} \left[\frac12 + B_0\right] \xi DB, ~~  \delta B =  \bar{\tau} \sqrt{\frac{|\tau|}{2}}  \left[\frac12 + B_0 - \frac{2}{iC_0|\tau|}   \right] \xi DB, &  \label{Case1Bderivatives} \\ 
& \Delta B =  e^{-2a-2\ba} \left[ \frac{|\tau| G_{,u}}{C_0}e^{a+\ba} \xi  - (f+\bar{f}) + \frac{v^2 |\tau|^3}{2}+\frac{G_{,uu}}{2C_0|\tau|}  \right]DB &, \end{aligned} \eeq
\noindent where $B_0$ is the following complex rational function of $\tau$, 
\beq &B_0 =  \frac{2}{iC_0 |\tau|} + \frac{C_0^2|\tau|}{C_0^2|\tau|-iC_0} - \frac{6C_0^2|\tau|^2}{3C_0^2|\tau|^2-2}   . & \nonumber \eeq 
\noindent Combining these functions, it is clear that both $\mu'$ and $\lambda'$ are expressed entirely in terms of $\tau$ and $\xi$. 

At this point we are able to prove that at second order, at least two functionally independent invariants are produced if we wish to produce a vacuum Kundt wave with $q \geq 4$ in the algorithm. 
\begin{lem} \label{lem:AclassInvCount}
All vacuum Kundt waves with the metric function $f(\z,u)$ of the form \eqref{AclassMetric} and an invariant count starting with $(0,1,2,...)$ in the Karlhede algorithm must end at third order; i.e., with an invariant count $(0,1,2,2)$.  
\end{lem}
\begin{proof}
The last invariant given in $(D^2 \Psi)_{51';11'} $ gives one new candidate for a functionally independent invariant: $\frac54 \tau \nu' + \bar{\nu'}\alpha$. Applying the transformation law for $\nu'$ it is seen that we may remove the majority of the terms in $\nu'$ and instead study the new invariant:  $\frac54 \tau( \nu + \Delta \bar{B}) + \alpha  (\bar{\nu} + \Delta B).$ As $|\alpha| \neq \frac54 |\tau|$, we may always combine this and its conjugate to produce a simpler invariant  \beq \tilde{\nu}=  \nu + \Delta \bar{B}. \label{NewNu} \eeq

 Denoting $F_x = Re(f_1)$ and $F_y = Im(f_1)$ we remove those terms in $\tilde{\nu}$ that are functionally dependent on $\tau$ and its conjugate to produce a new invariant:
\beq &N = \sqrt{\frac{\bar{\tau}}{\tau}}(\bar{f}_1|\tau|^{-1}e^{-2a-2\ba} - N'_0)|\tau| + D\bar{B} (N'_1- N'_0)& \nonumber \\
&N'_0 = \left( F_x (\z+\bz) + i F_y(\z-\bz) + 2Re(f_2) \right) e^{-2a-2\ba}, & \label{AclassN} \\ 
&N'_1 = \left[ \frac{|\tau|G_{,u}}{C_0}\left( v+ \frac{G_{,u}}{C_0|\tau|^2} \right) + \frac{v^2 |\tau|^3}{2}+\frac{G_{,uu}}{2C_0|\tau|} \right]e^{-2a-2\ba}. & \nonumber \eeq
\noindent Multiplying $\sqrt{\frac{\bar{\tau}}{\tau}}=e^{a-\ba}$ to $N$ and taking the difference of this new quantity with its conjugate,
\beq  e^{a-\ba} N - e^{\ba-a}\bar{N} =  -\frac{4iC_0|\tau|^2(N_1'-N_0')}{3C_0^2|\tau|^2-2} - 2iF_ye^{-2a-2\ba} , \nonumber \eeq
\noindent then by removing this term from $N$ leaves 
\beq N'_2 = (F_x|\tau|^{-1}e^{-2a-2\ba} - N_0')|\tau| + C_0|\tau|F_y e^{-2a-2\ba}.  \label{OiveyInv} \eeq 
\noindent We calculate the triple wedge product of this invariant with the previous invariants. The coefficients of the triple wedge product relative to the coordinate 3-form basis are extensive. However, only one is necessary if we wish that the triple wedge product vanishes, the vanishing of the $da \wedge d\ba \wedge dv$ coefficient yields
\beq -e^{-3a-3\ba}[4(-C_0F_y + iF_y(\z-\bz) + 2Re(f_2)) + 2C_0F_y] = 0. \nonumber \eeq 
\noindent As $\z-\bz$ is a linear function in $a+\ba$, $F_y$ must vanish and hence $Re(f_2) = 0$ as well. 

These constraints cause  $(F_x|\tau|^{-1}e^{-2a-2\ba} -N'_0)$ to vanish and so we work with the remaining invariant $N'= N'_1-N'_0 = (N'_1 - F_x|\tau|^{-1}e^{-2a-2\ba})|\tau|^{-3}$,  
\beq N' = \left[ \frac{G_{,u}}{C_0|\tau|^{2}}\left( v+ \frac{G_{,u}}{C_0|\tau|^2} \right) + \frac{v^2}{2}+\frac{G_{,uu}-2C_0F_x}{2C_0|\tau|^4}  \right]e^{-2a-2\ba}. \nonumber \eeq
\noindent Using the same procedure of equating the triple wedge product of $\ba-a$, $\xi$ and $N'$, we examine the $da \wedge du \wedge dv$-component, equating this to zero we find a differential equation: \beq  2G_{,uu}G_{,u} + C_0G_{,uuu} - 2C_0^2 F_{x,u}=0. \nonumber \eeq

\noindent Integrating we find that $F_x = \frac{G_{,uu}}{2C_0} + \frac{G_{,u}^2}{2C_0^2}+C_2$ and hence $N' = \frac{\xi^2}{2} + \left[ \frac{C_2}{|\tau|^4} \right] e^{-2a-2\ba}$

To continue, we eliminate the parts of this invariant expressed in terms of previous invariants, by denoting $N'' = N' - \frac{\xi^2}{2}$. We take the triple wedge product of this invariant with $a-\ba$ and $\xi$ to produce the following equation in the $da \wedge d \ba \wedge du$ component which must vanish: $G_{,uu}= 0$.
\noindent Denoting $G_{,u} = C_2$, the remaining invariant becomes $N'' =  C_2 e^{-2a-2\ba}(C_0^2|\tau|^4)^{-1}$. If we wish to have only two functionally independent invariants at second order, $C_2 =0$. This is generically the case;  if $G_{,u} \neq 0$ 
\noindent we  may always set $G=C_2u+C_3$ to zero using the coordinate transformation, \eqref{KWcoordtransf} of the form:
$u' = h(u),~~v' = \frac{v}{h_{,u}} + \frac{h_{,uu}}{2h_{,u}^2\tau|^2},~~h_{,u} = e^{-\frac{2}{C_0} G}$.
\noindent Applying this transformation, the analytic  function $f(\z,u)$ becomes
\beq  f'(\z,u) &=& \frac{C_0^2}{16} e^{\frac{-4i(\z+C_1)}{C_0}}. \nonumber \eeq

At third order there are no candidates for a third functionally independent invariant as all frame derivatives of $\xi$ produce invariants expressed in terms of the previous invariants.
\noindent The Karlhede algorithm terminates with an invariant count $(0,1,2,2)$. 
\end{proof}

Thus we have shown that vacuum Kundt waves with an invariant count of $(0,1,2,3,4,4)$ in the Karlhede algorithm cannot occur as the metrics with invariant counts starting with $(0,1,2,...)$ must have $(0,1,2,2)$ at the next order. 

\end{subsection}
\begin{subsection}{Vacuum Kundt waves with $(0,2,2,...)$} 
To begin, we prove a more general result for the vacuum Kundt waves with invariant count $(0,n,...)$ $1\leq n \leq4$ and $|\alpha| = 
\frac54 |\tau|$. 
\begin{lem} \label{lem:Collins}
 For those vacuum Kundt waves with at least one functionally independent invariant appearing at first order and such that $|\alpha| = \frac54 |\tau|$ then $\bar{\alpha} \neq e^{i \theta} \frac54 \tau$, $\theta \in \mathbb{R}$.  
\end{lem}
\begin{proof}
Expanding the conjugate of $\alpha$ using \eqref{SpinCoef1} we find 
\beq \frac{\tau}{4} + e^{\ba-a} \bz^{-1}_{,\ba} = - \frac54 e^{i \theta} \tau. \nonumber \eeq
\noindent Upon simplification this leads to the equation 
\beq \frac{1+5e^{i\theta}}{4} \bz_{,\ba} =  \z(a,u)+ \bz(\ba,u). \nonumber \eeq
\noindent A contraction arises here, as we may differentiate with respect to $a$ giving $\z_{,a} = 0$. 
\end{proof}

\noindent Recalling the comment after equation \eqref{NullRot} there are three cases to consider depending on the phase of the conjugate of $\alpha$. This lemma implies that the first two cases where $\bar{\alpha} = \pm \frac54 \tau$ cannot occur. 

To start narrowing the possibilities for $f(\z,u)$ we consider the wedge products of invariants built out of $\alpha$, $\tau$ and their conjugates. As $T=e^{\ba-a} = \sqrt{\tau/ \bar{\tau}}$, and $A = \z_{,a} = (e^{a-\ba}(\bar{\alpha} - \tau))^{-1}$ are both invariants it will be helpful to consider the triple wedge product: 
\beq & dA \wedge d\bar{A} \wedge dT = -T(\bz_{,\ba \ba} \z_{,a u} - \bz_{,\ba u} \z_{,a a}) da \wedge d\ba \wedge du.& \label{AbAT} \eeq
\noindent Alternatively, using the invariant $M = |\tau|^{-1} = \z(a,u)+\bz(\ba,u)$, we have another equation as the coefficient of the triple wedge product:
\beq &dT \wedge dM \wedge dA = -T(\bz_{,\ba \ba} \z_{,u} +  \bz_{, \ba \ba} \bz_{,u} - \bz_{,\ba u} \bz_{,\ba} - \bz_{,\ba u} \z_{,a}) da \wedge d\ba \wedge du.& \label{TNA} \eeq
\noindent Equating these two wedge products to zero, we have sufficient information to solve for $f(\z,u)$ in the vacuum Kundt wave metrics with invariant count $(0,2,...)$, and hence narrow down the possibilities for those spacetimes with invariant count $(0,2,2,...)$. 

\begin{lem}\label{lem:Bclass}
 The vacuum Kundt wave metrics for which the triple wedge product of $\alpha$, $\tau$ and their conjugates vanish have the following form:
\beq \tilde{f}(\z,u) &=& - \frac{F(u)^2}{16}e^{\frac{4(\z-f_0(u))}{iF(u)}} + g(u) \z + g_0(u) \label{BAmetric} \\
\tilde{f}(\z,u) &=& \frac{c^2}{16} e^{\frac{4(\z-f_1(u))}{c}} + g_1(u) \z + g_2(u),~~Re(C)\neq 0 \label{BBmetric} \\
\tilde{f}(\z,u) &=& f_2(\z-c_0 - i F_3(u))+g_3(u)\z + g_4(u) \label{BCmetric} \eeq
\end{lem}
\begin{proof}
Equating equations \eqref{AbAT} and \eqref{TNA} to zero, we have two differential equations for $\z(a,u)$ and its conjugate. There will be four cases depending on whether $\z_{,aa}$ and $\z_{,ua}$ are zero or not. 

\noindent {\bf Case 1} - $ \z_{,aa} = 0,~\z_{,ua} \neq 0$:

 Equation \eqref{AbAT} vanishes entirely while  \eqref{TNA} implies $\z_{,a} = -\bz_{,\ba}$, so that \beq \z = iF(u) a + f_0(u) \label{BAzed} \eeq \noindent Solving for $a$ and integrating  $\tilde{f}_{,\z\z} = e^{4a}$ we find the form \eqref{BAmetric}.

\noindent {\bf Case 2} - $ \z_{,aa} = 0,~\z_{,ua} = 0$:

 Here the constraints immediately imply \beq \z = ca + f_1(u). \label{BBzed} \eeq \noindent Solving for $a$ and integrating $\tilde{f}_{,\z\z} = e^{4a}$ yields the analytic function \eqref{BBmetric}.

\noindent {\bf Case 3} - $ \z_{,aa} \neq 0,~\z_{,ua} = 0$:

  These assumptions cause \eqref{TNA} to become $\z_{,u} + \bz_{,u} = 0$, implying $\z$ takes the form: \beq \z = \dot{f}_2^{-1}(a) + i F_3(u) + C_0 \label{BCzed}. \eeq \noindent Solving for $a$ and assuming $\dot{f}_2 = \frac14 ln \ddot{f}_2$, the expression $\tilde{f}_{,\z\z} = e^{4a}$ becomes, \beq \tilde{f}_{,\z\z} = \ddot{f}_2(\z-C_0-iF_3). \nonumber \eeq \noindent As $\dot{f}_2$ and $\ddot{f}_2$ are arbitrary functions of $u$, we make one more assumption, $\ddot{f}_2 = f_{2,\z\z}$. Integrating twice yields the desired metric function \eqref{BCmetric}. 

\noindent {\bf Case 4} - $ \z_{,aa} \neq 0,~\z_{,ua} \neq 0$

Re-arranging the functions we find \beq \frac{\z_{,au}}{\z_{,aa}} = \frac{\bz_{,\ba u}}{\bz_{,\ba \ba}} \nonumber \eeq \noindent which is equivalent to $\z_{,au} - F_5(u) \z_{,aa} =0 $. Integrating with respect to $a$ yields $\z_{,u} - F_5(u) \z_{,a} = f_4(u)$. Substituting this into \eqref{TNA} we find that $f_4 = iF_4$, so that $\z$ takes the form \beq \z = \dot{f}_6^{-1}\left(a + \int F_5 du \right) + i \int F_4 du \label{BDzed}. \eeq \noindent Solving for $a$ and assuming $\dot{f}_6 = \frac14 ln \ddot{f}_6 $ and $\ddot{f}_6 = f_{6,\z\z}$, we integrate twice to find 
\beq \tilde{f}(\z,u) &=& e^{-\int F_5(u) du} f_6 \left(\z - i \int F_4(u) du \right) + g_5(u) \z + g_6(u). \nonumber \eeq 

Any metric with a function of this form may be transformed into one of the form \eqref{BCmetric} using the transformation 
$ u' = h(u),~~v' = \frac{v}{h_{,u}} - \frac{h_{,uu}}{2h_{,u}^2 |\tau|^2},~~ h_{,u} = e^{-\frac{\int F_5 du}{2}}.$
\noindent The division of the case with $\z_{,aa} \neq 0$ cannot be made by $\z_{,au}$ vanishing or not; it is a coordinate-dependent distinction. 
\end{proof}

These metrics do not yet belong to the $(0,2,...)$ class as we must determine whether $\gamma$ may be set to zero or not. If $\gamma$ is non-zero, the various triple wedge products involving $\gamma$ with $\alpha$, $\tau$ and their conjugates give further conditions on the metric function $f(\z,u)$. By lemma \ref{lem:Collins} we see that $\bar{\alpha} \neq \pm \tau$ and hence we may eliminate the real or imaginary part of $\gamma$ but not both as the ratio of the real part to the imaginary part of the quantity, $\alpha B + \frac{5}{4} \bar{B} \tau$, is $tan[\frac12(arg(\alpha)+arg(\tau))] = tan(arg(e^{iC})) = C \neq 0$ \cite{Collins91}. 

Opting to eliminate the real part of $\gamma$, we note that the purely imaginary invariant $\gamma'$ is invariant under any null rotation preserving $Re(\gamma)=0$. Thus, without fixing the frame any further, the transformed scalar $\gamma'$ is a Cartan invariant:
\beq \gamma' &=& i(Im(\gamma) - CRe(\gamma)) \nonumber \\
&=& i \frac{\sqrt{|\tau|}}{2\sqrt{2}} \left[ \frac{i\z_{,u}}{\z_{,a}} - \frac{i\bz_{,u}}{\bz_{,\ba}} + C \left( |\tau|^2v +  \frac{\z_{,u}}{\z_{,a}} + \frac{i\bz_{,u}}{\bz_{,\ba}}  \right)     \right]e^{-a-\ba}. \nonumber \eeq 
\noindent  We may consider the triple wedge product of the differentials of three invariants constructed from $\gamma'$, $\tau$, $\alpha$ and their complex conjugates to determine the feasibility of the vacuum Kundt waves with invariant count $(0,2,2,3,4,4)$.

\begin{lem} \label{lem:BclassInvCount} 
 The class of vacuum Kundt waves with an invariant count beginning with (0,2,2,...) and $|\alpha| = \frac54 |\tau|$ cannot occur. 
\end{lem}

\begin{proof}
Taking the triple wedge product of the invariants $e^{\ba-a}$, $\z_{,a}$ and $\gamma'$, we examine the coefficients of $da \wedge d\ba \wedge dv$,  $da \wedge du \wedge dv$ and  $d\ba \wedge du \wedge dv$, equating these coefficients to zero we find two constraints:
\beq & \frac{i C|\tau|^{\frac52} e^{-2a}}{2\sqrt{2}} \z_{,aa} = 0,~~ \frac{i C|\tau|^{\frac52} e^{-2a}}{2\sqrt{2}} \z_{,au} = 0. & \nonumber \eeq
\noindent Immediately we see that the metric function must be of the form \eqref{BBmetric} with the corresponding form of $\z(a,u)$ given in \eqref{BBzed}. Expressing $\alpha$ and $\tau$ in terms of this function the required equality $|\alpha| = \frac54 |\tau|$ implies
\beq \frac{24 |c|^2}{16} - \frac{C_0 |\tau|^{-1}}{2} - |\tau|^{-2} = 0, \nonumber \eeq
\noindent  where $|\tau|^{-1} = C_0 (a+\ba)+iC_1(a-\ba)+2Re(f_0)$ with either $C_0$ or $C_1$ non-zero. Expanding $|\tau|^{-1}$  and differentiating twice with respect to $a$ we find a constant that must vanish:
\beq C_0 + i C_1 = 0. \nonumber \eeq
\noindent This produces a contradiction as we have assumed $\z_{,a} \neq 0$, thus there are no vacuum Kundt wave spacetimes with an invariant count $(0,2,...)$ where the first order  Cartan invariants satisfy $|\alpha| = \frac54 |\tau|$. 
\end{proof}

We have shown that the collection of vacuum Kundt waves must have either an invariant count $(0,2,2)$ with all isotropy fixed at first order, or an invariant count of $(0,2,3...)$ with $|\alpha| = \frac54 |\tau|$ implying all isotropy is fixed at second order. Regardless of either case, none of the potential spacetimes arising from these subclasses  produce an invariant count with $q=5$.  This lemma completes the proof of theorem \ref{thm:4itis} as we have shown the two possibilities for the Karlhede algorithm  requiring $q=5$ cannot occur. 
\end{subsection}
\end{section}

\begin{section}{Sharpness of the $q\leq 4$ Upper Bound} \label{Sharpz}
In this section we will show that the new upper bound is indeed sharp by producing an explicit metric function $f(\z,u)$. 

\begin{thm} \label{Aclassq4} 
The vacuum Kundt waves with invariant count $(0,1,3,4,4)$ are of the form 
\beq f(\z,u) = \frac{C_0^2}{16} e^{\frac{-4i(\z+C_1)}{C_0}} + f_1(u) \z + f_2(u) \nonumber \eeq
\noindent where $f_1$ and $f_2$ are non-constant and satisfy either \beq &f_1 = (C_2+i)F_y,~~Re(f_2) = C_3 + \frac{C_0}{2} ln(F_y) F_y,~~F_y \neq Cu^{-2},  & \nonumber \\ &or & \label{AclassG1c} \\ &f_1 = F_x,~~Re(f_2) = C_3F_x,~~F_x \neq Cu^{-2}. \nonumber  & \eeq 
\end{thm}

\begin{proof}
To prove this fact, we calculate the quadruple wedge product of the differentials of $a-\ba$, $\xi$ and two new invariants arising in $N$ where the invariant $N_2'$ in \eqref{OiveyInv} is now denoted as $N_0$, and $N_1$ arises from the imaginary part of $N$,

\beq &N = \sqrt{\frac{\bar{\tau}}{\tau}} \left( N_0|\tau| + DB(DB+D\bar{B})^{-1}\left( N_0 + \frac{|\tau|^2 \xi^2}{2} + N_1 \right) \right) & \nonumber \\
&N_0 = \left( -C_0F_y + iF_y(\z-\bz) + 2Re(f_2)\right) e^{-2a-2\ba}, & \label{AclassNSharp} \\ 
&N_1 = \left[ \frac{(G_{,u}^2 + G_{,uu}C_0 - 2C_0^2F_x)}{2C_0^2|\tau|} + \frac{[C_0|\tau|^2-1]F_y}{C_0|\tau|^2} \right]e^{-2a-2\ba}. & \nonumber \eeq

\noindent In these coordinates, the invariants are a bit complicated; one may make a coordinate transformation to remove $G(u)$ in the function $f(\z,u)$ in \eqref{AclassMetric}. Applying the transformation \eqref{KWcoordtransf}: $u' = h(u),~~v' = \frac{v}{h_{,u}} + \frac{h_{,uu}}{2h_{,u}^2\tau|^2},~~h_{,u} = e^{-\frac{2}{C_0} G}$.
\noindent  Relabeling the arbitrary functions $f_1$ and $f_2$, the analytic  function $f(\z,u)$ becomes
\beq & f'(\z,u) = \frac{C_0^2}{16} e^{\frac{-4i(\z+C_1)}{C_0}} + f'_1(u') \z + f'_2(u'). & \label{AclassNoG} \eeq

Dropping the primes and repeating the calculations in lemma \ref{lem:Aclass} and subsection \eqref{Aclassq4upb} with this new function, one finds that $N_0$ and $N_1$ are now
\beq \begin{aligned} &N_0 = \left( -C_0F_y + iF_y(\z-\bz) + 2Re(f_2)\right) e^{-2a-2\ba}, & \label{AclassNSharpNoG} \\ 
&N_1 = \left[ \frac{- 2C_0^2F_x}{2C_0^2|\tau|} + \frac{[C_0^2|\tau|^2-1]F_y}{C_0|\tau|^2} \right]e^{-2a-2\ba}. &  \end{aligned} \eeq
\noindent From the quadruple wedge product $d(a-\ba) \wedge d \xi \wedge dN_0 \wedge dN_1$ we find the sole coefficient 
\noindent yields three essential equations whose vanishing is necessary and sufficient for the 4-form to vanish:
\beq & (Re(f_2) +\frac{C_0}{4} F_y) F_{x,u} - Re(f_2)_{,u} F_x = 0 \nonumber \\
& (Re(f_2) +\frac{C_0}{4} F_y) F_{y,u} - Re(f_2)_{,u} F_y = 0 & \nonumber \\
& F_y F_{x,u} - F_{y,u}F_x = 0. \nonumber \eeq

To solve these equations we must consider two cases depending on whether $F_y =0$ or not. In the case that $F_y$ does vanish, we find that $Re(f_2)$ may be expressed in terms of derivatives $F_x$, an arbitrary function: 
\beq &Re(f_2) = C_3F_x.& \label{G0SharpA} \eeq
\noindent While if $F_y \neq 0$ and arbitrary, we find that 
\beq F_x = C_2F_y,~Re(f_2) = [C_3 +\frac{C_0}{4} ln(F_y)]F_y. \label{G0SharpB} \eeq
\noindent The choice of these functions is reflected in the structure of the invariants. Supposing that $F_y = 0$, we may express $N_1$ in terms of $N_0 = Re(f_2)e^{-2a-2\ba}$, 
\beq &N_1 = \left[ \frac{C_3}{|\tau|} \right] N_0.& \nonumber \eeq
\noindent While if $F_y \neq 0$ we find that $N_0$ and $N_1$ may be expressed in terms of $N_2 = F_y e^{-2a-2\ba}$,
\beq \begin{aligned} & N_0 = N_2(C_0|\tau|^{-1}+2C_1 + ln(N_2/2)),~~ N_1 = \left[ \frac{C_2}{|\tau|} + \frac{C_0^2|\tau|^2-2}{C_0|\tau|^2} \right] N_2.& \label{q4A} \end{aligned} \eeq

Regardless of whether $F_y \neq 0$ or not, the third second order invariant arising here is of the form
\beq \tilde{N} = F_0(u)e^{-2a-2\ba}. \nonumber \eeq
\noindent The frame derivatives of this invariant produce only one new functionally independent invariant,
\beq \sqrt{\frac{2}{|\tau|}} \Delta \tilde{N} = F_{0,u} e^{-3a-3\ba} \nonumber \eeq
\noindent To determine the full class of $(0,1,3,4,4)$ vacuum Kundt waves, we must avoid those functions $F_0$ which give the invariant count $(0,1,3,3)$, this can only happen if $F_0$ is constant or when it satisfies the following differential equation,
\beq F_{0,u} = -2 \sqrt{C^{-1}_4} F_0^{\frac32} \nonumber \eeq
\noindent by integrating one finds that $F_0 = C_4 u^{-2}$. 

In the case that $F_0$ is constant, all of the metric functions in \eqref{AclassNoG} are independent of $u$ and hence this is a $G_1$ metric with no $u$-dependence. In the other case, we may make a coordinate transformation: $u' = h(u),~~v' = \frac{v}{h_{,u}} + \frac{h_{,uu}}{2h_{,u}^2\tau|^2},~~h_{,u} = u^{-1}.$
\noindent Dropping the primes, in these new coordinates the $(0,1,3,3)$ metrics with $F_y =0$ are now of the form 
\beq & f(\z,u) = \frac{C_0^2}{16} e^{\frac{-4i(\z+\frac{iC_0u}{2}+C_1)}{C_0}} + C_2 \z + C_3& \label{Aclass:G1a} \eeq
\noindent while those metrics with $F_y \neq 0$ are
\beq & f(\z,u) = \frac{C_0^2}{16} e^{\frac{-4i(\z+\frac{iC_0u}{2}+C_1)}{C_0}} + C_2\z + iC_3\left(\z+\frac{iC_0u}{2}\right) + C_4.& \label{Aclass:G1b} \eeq
\noindent From \cite{SGP} we conclude these are all  $G_1$ spacetimes.
\end{proof}

\noindent We conclude this section with the result that the sharpness of the upper bound has been confirmed.

\end{section}

\begin{section}{Uniqueness of the Vacuum Kundt Waves with $q=4$ in the Karlhede Algorithm}

From the invariant count trees in Appendix D, the remaining possibilities for vacuum Kundt waves to attain $q=4$ outside of the class with $(0,1,3,4,4)$ are: $(0,1,2,3,3), (0,1,2,4,4), (0,2,2,3,3), (0,2,2,4,4),(0,2,3,4,4),$ and $(0,3,3,4,4)$. It was proven in section \ref{Leq5}  that the first four cases cannot occur due to lemma \ref{lem:AclassInvCount} and lemma \ref{lem:Bclass}, respectively. The last case may be ignored by applying lemma \ref{lem:Noq4Cclass} to show that the vacuum Kundt waves with invariant count $(0,3,3,4,4)$ cannot occur.  
Thus we need only investigate the existence of the $(0,2,3,4,4)$ vacuum Kundt waves to determine the uniqueness of the vacuum Kundt waves with $q=4$.  In this case  $\gamma$ may be set to zero and the invariant coframe is entirely fixed. To continue, we examine the second order invariants arising from the frame derivatives \eqref{Icoframe} applied to the simpler set of invariants: \beq &a-\ba = \frac12 ln \left(\frac{\bar{\tau}}{\tau}\right), \z_{,a} = \sqrt{\frac{\bar{\tau}}{\tau}} (\bar{\alpha}-4\tau)^{-1},~~ \z+\bz=|\tau|^{-1}.& \nonumber \eeq \noindent By direct calculation we may prove the following proposition. 
\begin{prop} \label{prop:KW2ndOInv}
 For all vacuum Kundt waves with $|\alpha| \neq \frac54 |\tau|$, the second order Cartan invariants with no functional dependence on the previous invariants consist of the spin-coefficients $\mu'$, $\lambda'$, $\nu'$:  and the frame derivatives: \beq &\sqrt{\frac{|\tau|}{2}}M_0 = |\z_{,a}|^2 \Delta (a-\ba),~~\sqrt{\frac{|\tau|}{2}}M_1 =\z_{,a} \Delta \z_{,a},& \nonumber \\ &\sqrt{\frac{|\tau|}{2}}M_2 = \Delta (\z+\bz),  \sqrt{\frac{\tau}{\bar{\tau}}} M_3 = \z_{,a} \bar{\delta} \z_{,a}. & \nonumber \eeq 
\noindent In the coordinate system with $a=\frac14 ln(f_{,\z\z})$, these invariants take the form:
\beq   M_0 &=&  - e^{-a-\ba} ( \z_{,u}\bz_{,\ba} - \bz_{,u}\z_{,a} ) + \frac{\tau}{\bar{\tau}} B' \bz_{,\ba} - \frac{\bar{\tau}}{\tau} \bar{B}'\z_{,a},   \nonumber \\
 M_1 &=&  e^{-a-\ba} (\z_{,au}\z_{,a}-\z_{,u} \z_{,aa} ) + \frac{\tau}{\bar{\tau}} B' \z_{,aa},   \nonumber \\
 M_2 &=& \frac{\tau}{\bar{\tau}} B'+ \frac{\bar{\tau}}{\tau} \bar{B}', \nonumber \\ 
 M_3 &=& \z_{,aa},  \nonumber \\
 \lambda' &=& \frac{\bar{B} \bar{\tau}}{2} + \sqrt{\frac{\bar{\tau}}{\tau}} \frac{2\bar{B}}{\bar{\z}_{,\ba}} + \bar{B} \pi+ \bar{B} D\bar{B} + \bar{\delta} \bar{B},  \nonumber \\ 
 \mu' &=& \frac{\bar{B} \tau}{2} + B \pi +  B D\bar{B} + \delta \bar{B},  \nonumber \\
 \nu' &=& \nu + 2\bar{B} \gamma + \frac32 \bar{B}^2 \tau + B\bar{B}(\pi + 2\alpha) + \Delta \bar{B} + \bar{B} \delta \bar{B} + B \bar{\delta} B + B \bar{B} D\bar{B}  \nonumber \eeq 

\noindent where the unprimed spin-coefficients are defined in \eqref{SpinCoef1} and $B'=\sqrt{\frac{\bar{\tau}}{\tau}}\sqrt{\frac{2}{|\tau|}}B$ where $B$ is the null rotation parameter setting $\gamma' = 0$ in \eqref{NullRot}
\beq &B' = e^{-a-\ba} \left[ DB' v - \frac54 \left( \frac{\bz_{,u}}{\bz_{,\ba}} - \frac{\z_{,u}}{\z_{,a}} \right) + \frac{\bz_{,u}}{\bz_{,\ba} |\tau|^2} DB' \right], & \nonumber \\ 
& DB' = \frac{16|\tau|^2}{25|\tau|^2-16|\alpha|^2} \left( |\tau| + \frac{1}{\bz_{,\ba}} \right). & \nonumber \eeq   
\end{prop}

To inquire into the uniqueness of the $q=4$ vacuum Kundt waves, we classify the vacuum Kundt waves with invariant count beginning with $(0,2,3,...)$. In order to identify this subclass we consider the quadruple wedge products of $ d(a-\ba) \wedge d\z_{,a} \wedge d M_{i} \wedge d M_{j}$ and $d(a-\ba) \wedge d |\tau|^{-1} \wedge d M_{i} \wedge d M_{j}$. If there are only three functionally independent invariants at second order, all twelve quadruple wedge products must vanish.
\begin{lem} \label{lem:BAmetric} 
 The vacuum Kundt wave metrics with an analytic function of the form \eqref{BAmetric}: \beq &\tilde{f}(\z,u) = - \frac{F(u)^2}{16}e^{\frac{4(\z-F_0(u))}{iF(u)}} + g(u) \z + g_0(u)& \nonumber \eeq
\noindent have the invariant count $(0,2,4,4)$.  
\end{lem}
\begin{proof}
To start, we make a coordinate transformation to remove the imaginary part of $f_0$ in \eqref{BAmetric} via the transformation $ u' = h(u),~~v' = \frac{v}{h_{,u}} - \frac{h_{,uu}}{2h_{,u}^2 |\tau|^2},~~ h_{,u} = e^{-\frac{2Im(f_0)}{F(u)}}.$
\noindent Writing $\z(a,u) = iF(u)a+F_0(u)$, we find that the first order invariants arising from $\tau$ and $\alpha$ are 
\beq \frac12 ln(\bar{\tau}/\tau) = a-\ba,~~\z_{,a} =i F(u),~~|\tau|^{-1} = iF(a-\ba)+F_0.  \label{BAmetricInv} \eeq
\noindent As $F' \neq 0$ in order to avoid metrics of the form \eqref{BBmetric}, we take its inverse locally and express all other functions of $u$ in terms of it. \beq F_0 = \mathfrak{F}_0(F). \nonumber \eeq
\noindent Thus we are left with $a-\ba$ and $\z_{,a}$ as invariants. Noting that $M_1 = M_1' + M_2$, where $M_1$ is  \beq M_1' = e^{-a-\ba} F F_{,u} = e^{-a-\ba} \mathfrak{F}(F). \nonumber \eeq
\noindent Removing the $u$-dependent piece, we may solve for $a-\ba$ as a third functionally independent invariant. Taking $M_2$ in proposition \ref{prop:KW2ndOInv} we eliminate all terms dependent on $a,\ba$ and $u$ leaving $v$ as the last invariant at second order to complete the set $\{ a-\ba, a+\ba, F(u),v\}$ with the spin-coefficients at first and second order acting as the classifying manifold along with the frame derivatives of $v$ and $a+\ba$.
\end{proof}

With this case eliminated, we may consider those vacuum Kundt waves with $\z_{,au} = 0$. The vanishing of the quadruple wedge products produce six equations \beq &M_{i,u} M_{j,v} - M_{i,v} M_{j,u} = 0.& \label{4wwij} \eeq

\begin{lem} \label{lem:BBmetric} 
 The vacuum Kundt wave metrics with analytic function of the form \eqref{BBmetric}: \beq &\tilde{f}(\z,u) = \frac{c^2}{16} e^{4\left( \frac{\z}{c}-\frac{iF_1(u)}{|c|^2}\right)} + g_1(u) \z + g_2(u),~~Re(c)\neq 0,& \nonumber \eeq
\noindent have the invariant count $(0,2,4,4)$ except in the subclass of these metrics with  
\beq &\tilde{f}(\z,u) = \frac{c^2}{16} e^{\frac{4(\z-C_0-iC_1u)}{c}} + c_2\z + Im(c_2)C_1u+C_3& \nonumber \eeq   
\noindent which have the invariant count $(0,2,3,3)$
\end{lem}
\begin{proof}
 We first examine the possibility of invariant counts of the form $(0,2,3,...)$ using the metric function \eqref{BBmetric}. In this case the function is $\z(a,u) = ca+f_1(u)$, we find that the first order invariants arising from $\tau$ and $\alpha$ are 
\beq a-\ba,~~\z_{,a} =c,~~|\tau|^{-1} = Re(c)(a+\ba)+iIm(c)(a-\ba)+Re(f_1).  \label{BBmetricInv} \eeq
\noindent At second order, $M_3 = M_1 = 0$, thus there is only one quadruple wedge product giving constraints on the metric functions. Mutiplying $cM_2$ and adding it to $M_0$ gives a useful invariant 
\beq \frac{M_0'}{c+\bar{c}} = -\frac{e^{-a-\ba}}{c+\bar{c}} Im(\z_{,u}\bar{c})+  \frac{\tau}{\bar{\tau}}B'. \nonumber \eeq
\noindent To calculate the wedge product we scale $M_0'$ and $M_2$ and use the following quantities:  \beq &M_0'' = M_0' \frac{25|\tau|^2-16|\alpha|^2}{16|\tau|^2 (c+\bar{c})}~~M_2' = \frac{M_2}{\frac{\tau}{\bar{\tau}}DB' + \frac{\bar{\tau}}{\tau}D\bar{B}'}.& \nonumber \eeq 
\noindent Substituting into equation \eqref{4wwij} and differentiating the whole expression by $v$ we find  the simpler constraint, $|\tau|_{,u} = 0$, 
\noindent from which we find that $Re(f_1')=0$ implying that $a-\ba$ and $a+\ba$ are the only first order invariants. 

Returning to the original invariants $M_0'$ and $M_2$, substituting into equation \eqref{4wwij} and denoting $Im(f_1) = F_y$ we find that this becomes, 
\beq iF_{y,uu}\left[ \frac{\tau}{\bar{\tau}} DB' + \frac{\bar{\tau}}{\tau} D\bar{B}'+\frac{5(c+\bar{c})}{4|c|^2}(\frac{\tau}{\bar{\tau}} DB' + \frac{\bar{\tau}}{\tau} D\bar{B}') + \frac{c+\bar{c}|DB'|^2}{|c|^2|\tau|^2}\right]. \nonumber \eeq
\noindent As before, setting this equation to zero we find that $F_y = Im(f_1) = C_1u$. Substituting the form of $f_1= C_0 +iC_1u$ into $B'$ in proposition \ref{prop:KW2ndOInv}, one may solve for $v$ in $M_2$ as the last functionally independent invariant. 

The only new functionally independent invariant arises in $\nu$ in \eqref{InvSpinC}, as this is the only function retaining $u$-dependence. Due to the formula for $\nu$ in \eqref{SpinCoef1} we may work with the simpler invariant, 
\beq  V &=& |\tau|^{-1}g_1 + g_1\z + \bar{g}_1 \bz + 2Re(g_2) \label{Vnu} \eeq
\noindent Denoting $g_1 = G_x+iG_y$. Taking the wedge product $da \wedge d\ba \wedge dv \wedge dV$ and and equating this to zero we find that $g_1 = G_x + iIm(c_2)$, $Re(g_2) = Im(c_2)C_1u$ and $G_{x,u} = 0$ and so $g_1 = c_2 \in \mathbb{C}$. As all $u$-dependence has been removed from the invariants, it is clear this is a $G_1$ space; the classifying manifold consists of the first order and second order invariants in terms of $a,\ba$ and $v$ along with the frame derivatives of $v$.

In the $(0,2,4,4)$ case, we may replace the complex-valued $f_1$ in \eqref{BBmetric} with a real-valued function of $u$. To do so, we apply the following coordinate transformation
\beq & u' = h(u),~~v' = \frac{v}{h_{,u}} - \frac{h_{,uu}}{2h_{,u}^2 |\tau|^2},~~ h_{,u} = e^{-\frac{2}{|c|^2}(Re(f_1)Re(c)+Im(f_1)Im(c))}.& \nonumber \eeq 
\noindent Then by making the gauge transformation, $F_1 = -Re(f_1)Im(c)+Im(f_1)Im(c)$, we recover the desired form.
\end{proof}

\begin{lem} \label{lem:BCmetric} 
 The vacuum Kundt wave metrics with analytic function of the form \eqref{BCmetric}: \beq &\tilde{f}(\z,u) = f_2(\z-C - i F_3(u))+g_3(u)\z + g_4(u) & \nonumber \eeq
\noindent have the invariant count $(0,2,4,4)$ except in the subclass of these metrics with 
\beq \tilde{f}(\z,u) = f_2(\z-C - i C_0u)+c_1\z + Im(c_1)C_0u)+C_2  \nonumber \eeq   
\noindent which have the invariant count $(0,2,3,3)$.
\end{lem}
\begin{proof}
To start we determine the conditions for an invariant count of $(0,2,3,...)$. Here, the metric function is $\z(a,u) = Z(a)+C+iF_3(u)$ and we find that the first order invariants arising from $\tau$ and $\alpha$ are 
\beq \z_{,a},~~|\tau|^{-1} = \z(a)+\bz(\ba).  \label{BCmetricInv} \eeq
\noindent Locally we may take the inverse of $\z_{,a}$ to solve for $a$ and use it as an invariant. Similarly we may do this for the conjugate, and hence at first order $a$ and $\ba$ may be treated as invariants. At second order, $M_3 = \z_{,aa}$ gives no new information. If we define a new invariant $M_1'\ = M_1 M_3^{-1}$, $M_2 = M_1 + \bar{M}_1$ and $M_0 = \bz_{,\ba}M_1'-\z_{,a}\bar{M}_1'$, there is only one quadruple wedge product giving constraints on the metric functions. 

Taking the quadruple wedge product with $M_1'$ and its conjugate and substituting into equation \eqref{4wwij}, we obtain
\beq -iF_{3,uu}\left[ \frac{\tau}{\bar{\tau}}DB' + \frac{\bar{\tau}}{\tau}D\bar{B}'+\frac{5(\z_{,a}+\bz_{,\ba})}{4|\z_{,a}|^2}\frac{\tau}{\bar{\tau}}DB' + \frac{\bar{\tau}}{\tau}D\bar{B}' + \frac{\z_{,a}+\bz_{,\ba}|DB'|^2}{|\z_{,a}|^2|\tau|^2} \right]. \nonumber \eeq
\noindent As before, setting this equation to zero we find that $F_3 = C_0u$. Substituting $F_3(u)$ into $B'$ in proposition \ref{prop:KW2ndOInv} it is clear that one may peel away the terms and coefficients of the $v$-linear term in $M_2$ to produce $v$ as the last invariant. 

From  the remaining second order invariants, \eqref{InvSpinC} it is clear that the only new functionally independent invariant arises in $\nu$ as this is the only function retaining $u$-dependence. Denoting $g_3 = G_x+iG_y$ we may work with the simpler invariant $V$ in \eqref{Vnu}:
\beq  V &=& |\tau|^{-1}g_3 + g_3\z + \bar{g}_3 \bz + 2Re(g_4) \nonumber \eeq
\noindent Repeating the calculation of the wedge product  $da \wedge d\ba \wedge dv \wedge dV$ and equating this to zero we find that $g_3 = G_x+iIm(c_1)$, $Re(g_4) = Im(c_1)C_0u$ and$g_3 = c_1\in \mathbb{C}$. As all $u$-dependence has been removed from the invariants, it is clear this is a $G_1$ space, the classifying manifold consists of the first order and second order invariants in terms of $a,\ba$ and $v$ along with the frame derivatives of $v$. 
\end{proof}
\end{section}

\begin{section}{An Invariant Classification of Vacuum Kundt Waves}

In proving the sharpness of the lowered upper bound we exhausted all of the branches of the invariant-count tree starting with $(0,1,...)$. Employing the first order Cartan invariants, $\alpha$, $\tau$ and $\gamma$, we may eliminate several branches from the remaining invariant-count trees in figure \eqref{1KarlAlg}. In this section we examine the remaining possibilities to produce  figure  \eqref{KAOmega} which summarizes all possible invariant counts for the vacuum Kundt waves in the Karlhede algorithm.
 \begin{figure}[h] 
  \centering 
  \includegraphics[scale=0.35]{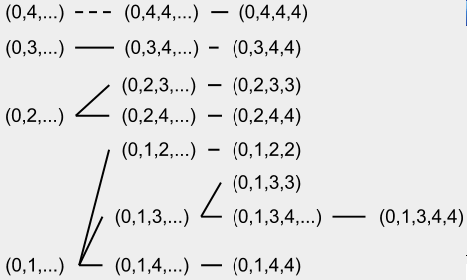} 
 \caption{All permissible invariant-count trees for the vacuum Kundt waves } \label{KAOmega}
\end{figure}

\begin{subsection}{Vacuum Kundt waves with $|\alpha| = \frac54 |\tau|$}

In the most general case, where a vacuum Kundt wave admits the following invariant counts, $(0,4,...)$, we may eliminate the scenario where $q=2$ by counting coordinates involved in the first order invariants.

\begin{lem} \label{lem:Noq2G0s}
All vacuum Kundt waves with invariant count $(0,4,...)$ must have $|\alpha| = \frac54 |\tau|$.  
\end{lem}
\begin{proof}
 Choosing Kundt coordinates, we calculate the quadruple wedge product of the differentials of the first order Cartan invariants $\alpha$, $\tau$ and their conjugates. As they are all functions of $a,\ba$ and $u$ relative to the special coordinate system, it is clear that 
\beq d \alpha \wedge d \bar{\alpha} \wedge d \tau \wedge d \bar{\tau} = 0. \nonumber \eeq
\noindent If the magnitudes of $\alpha$ and $\tau$ were not proportional we would always be able to set $\gamma = 0$, contradicting our assumption that four invariants appear at first order.   
\end{proof}

In the process of lowering the upper bound to $q\leq 4$  for the vacuum Kundt waves in the Karlhede algorithm and its sharpness, we have shown that those metrics with invariant count $(0,1,...)$ and $(0,2,...)$  must have $|\alpha| \neq \frac54 |\tau|$. Using the same approach we can demonstrate that the class of metrics with three functionally independent invariants must have $|\alpha| \neq \frac54 |\tau|$ also

\begin{lem} \label{lem:Noq4Cclass}
 For all vacuum Kundt wave spacetimes with invariant count $(0,3,...)$, the magnitude of $\alpha$ is never proportional to that of $\tau$; i.e., $|\alpha| \neq \frac54 |\tau|$. All remaining frame freedom is exhausted at first order by setting $\gamma =0$. 
\end{lem}

 \begin{proof}
Let us assume that the two magnitudes are equal, then by lemma \ref{lem:Collins}, $\alpha \neq \pm \tau$ and we may set either the real or imaginary part of $\gamma$ to zero. As before, we eliminate the real part of $\gamma$. The purely imaginary invariant $\gamma'$ is invariant under any null rotation preserving $Re(\gamma)=0$ due to the proportionality of the real and imaginary part of $\alpha B + \frac{5}{4} \bar{B} \tau$. Thus, without fixing the frame any further, the transformed scalar $\gamma'$ is a Cartan invariant:
\beq \gamma' &=& i(Im(\gamma) - C(a,\ba,u) Re(\gamma)) \nonumber \\
&=& i \frac{\sqrt{|\tau|}}{2\sqrt{2}} \left[ \frac{i\z_{,u}}{\z_{,a}} - \frac{i\bz_{,u}}{\bz_{,\ba}} + C \left( |\tau|v +  \frac{\z_{,u}}{\z_{,a}} + \frac{i\bz_{,u}}{\bz_{,\ba}}  \right)     \right]e^{-a-\ba}, \nonumber \eeq 
\noindent  and so we may consider the quadruple wedge product of the differentials of three invariants constructed from $\gamma'$, $\tau$,$\alpha$ and their complex conjugates: $|\tau|^{-1}$, $e^{\ba-a}$, $\z_{,a}$ and $\gamma'$. Doing so we find the sole coefficient of $da \wedge d\ba \wedge du \wedge dv$ is:
\beq \frac{ie^{2a}|\tau|^{\frac32} C}{2\sqrt{2}} \left(\bz_{,\ba \ba} \z_{,u} +  \bz_{, \ba \ba} \bz_{,u} - \bz_{,\ba u} \bz_{,\ba} - \bz_{,\ba u} \z_{,a}\right). \nonumber \eeq
\noindent If we wish to have three functionally independent invariants at first order, this equation must vanish.  However, this is exactly equation \eqref{TNA} used to determine the class of vacuum Kundt wave metrics with invariant count $(0,2,...)$. This contradicts our assumption and so $|\alpha| \neq \frac54 |\tau|$.   
\end{proof}
\noindent With this result we see that for all metrics with an invariant count $(0, n,...),~n<4$, we may always set $\gamma =0$ as $\frac54 |\tau| \neq |\alpha|$.
\end{subsection}
\begin{subsection}{Vacuum Kundt waves with $|\alpha| \neq \frac54 |\tau|$}

To complete the classification of the vacuum Kundt waves using the invariant counts arising from the Karlhede algorithm, we prove that the class of vacuum Kundt waves with invariant count $(0,3,3)$ and $(0,2,2)$ cannot occur in the next two lemmas. 

\begin{lem} \label{lem:No033} 
 If a vacuum Kundt wave spacetime admits three functionally independent invariants at first order of the Karlhede algorithm, it must belong to the $(0,3,4,4)$ class.
\end{lem}
\begin{proof} 
Supposing that we do have the invariant count $(0,3,3)$ we will show there is a contradiction. Denoting the triple wedge product $\Omega_3 = d(a-\ba) \wedge d \z_{,a} \wedge d(\z+\bz)$, we note \beq \Omega_3 = -(\bz_{,\ba \ba} (\z_{,u} +  \bz_{,u}) - \bz_{,\ba u} (\bz_{,\ba} + \z_{,a})) da \wedge d\ba \wedge du. \nonumber \eeq
\noindent We recall from equation \eqref{TNA} that if this equation vanishes only two functionally independent invariants appear at first order of the algorithm; thus this must be non-zero if we wish to have three invariants at first order. To impose the condition that no new functionally independent invariants appear at second order, we require the vanishing of all quadruple wedge products with $M_I$, $I=0,1,2,3,$
\beq \Omega_3 \wedge d M_I = -M_{I,v}(\bz_{,\ba \ba} (\z_{,u} +  \bz_{,u}) - \bz_{,\ba u} (\bz_{,\ba} + \z_{,a})) da \wedge d\ba \wedge du \wedge dv. \nonumber \eeq  
This can only occur if and only if $M_{i,v} = 0$, $i=0,1,2$. The $v$-coefficient of the first three $M_i$ yields two cases, depending on whether $\z_{,aa} = 0$ or not. 
\begin{itemize}
 \item If $\z_{,aa} \neq 0$, the vanishing of $\Omega_3 \wedge M_1$ implies $M_{1,v} = 0$; this can only occur if $DB = 0$ which is not possible, otherwise one would have $|\tau| = - \z_{,a}^{-1}$. If one were to impose this constraint, it immediately implies, $\z_{,a} = 0$ which cannot be true. 
\item If $\z_{,aa} = 0$, the vanishing wedge products $\Omega_3 \wedge M_0$ and $\Omega_3 \wedge M_2$ give the following equations \beq & \frac{\bar{\tau}}{\tau} D \bar{B}' \z_{,a} - \frac{\tau}{\bar{\tau}} DB'\bz_{,\ba}= 0, & \nonumber \\
& \frac{\tau}{\bar{\tau}}DB' + \frac{\bar{\tau}}{\tau} D\bar{B}' = 0.& \nonumber \eeq
\noindent  As $DB \neq 0$, we may solve one equation and substitute into the other, 
\beq \left[ \frac{\bz_{,\ba}}{\z_{,a}} +1  \right] D B \frac{\tau}{\bar{\tau}} = 0. \nonumber \eeq
\noindent This will only vanish if $\z_{,a} = - \bz_{,\ba}$; however, if this is the case, then \eqref{TNA} is satisfied and this spacetime belongs to the $(0,2,...)$ class, contradicting our assumption, and so it cannot occur.	
\end{itemize}
\end{proof}
\noindent Effectively we may differentiate those vacuum Kundt waves  with invariant count $(0,3,4,4)$ and $(0,4,4,4)$ by the non-vanishing of the first order invariant $|\alpha|-\frac54 |\tau|$. The Newman-Penrose field equations provide further classifying functions. 

We now consider the remaining branches of the vacuum Kundt waves with invariant count $(0,2,...)$, and show the subclass with invariant count $(0,2,2)$ cannot occur. 

\begin{lem} \label{lem:No022} 
 If a vacuum Kundt wave spacetime admits a two-dimensional isometry group it must belong to the $(0,1,2,2)$ class.
\end{lem}
\begin{proof}
 Supposing that only two functionally independent invariants appear at first order, we require that the wedge products of $d(a-\ba) \wedge d\z_{,a} \wedge dM_i$ and $d(a-\ba) \wedge d|\tau|^{-1} \wedge dM_i$ all vanish. 
\noindent If these wedge products are to vanish then either $\z_{,a}+\bz_{,\ba} = \z_{,u}+\bz_{,u} = \z_{,au} = \z_{,aa} = 0$ or $M_{i,v} = 0$. 

As in the proof of Lemma \ref{lem:No033} we may use the same argument for metrics \eqref{BAmetric}, and \eqref{BCmetric}  to show $M_{i,v} \neq 0$, $i=0,1,2$. In the case of the metric function \eqref{BBmetric} where $\z_{,aa} = 0$ and $\z_{,a} = -\bz_{,\ba} $, $M_{2,v} = 0$ occurs if and only if $\z_{,a} = 0$ which is not possible. If these wedge products do vanish, we must have $\z_{,a}+\bz_{,\ba} = \z_{,u}+\bz_{,u} = \z_{,au} = \z_{,aa} = 0$, implying that this metric belongs to the $(0,1,...)$ class.
\end{proof}

\end{subsection}

\end{section}

\begin{section}{Conclusions} \label{Treez}

In this paper we have invariantly classified all of the vacuum Kundt waves by exhaustively listing all invariant counts that appear as states in the Karlhede algorithm. Using the invariants produced by this method, we examine each invariant count to determine if the spacetime is integrable. In many cases whole branches do not occur or are significantly simplified; the results of this analysis are summarized in table form in the following two tables \eqref{table:KundtWavesG0} and \eqref{table:KundtWavesG1}.


This analysis was motivated by previous work on the upper bound of the Karlhede algorithm applied to type N spacetimes; it was conjectured that $q\leq 5$ \cite{Collins91, MRVickers}  for the vacuum Kundt waves; however, this upper bound was not shown to be sharp. We have lowered the upper bound to $q \leq 4$ and produced an example by integrating the class of vacuum Kundt waves with $(0,1,3,4,4)$ proving the sharpness of the bound. Furthermore, we have proved that this class is unique as it is the only class requiring the fourth derivative of the curvature to invariantly classify its members. 

\begin{table}[t] 
\begin{center} 
\begin{tabular}{cc}
Invariant Count & $f(\z,u)$    \\ [0.5ex]
\hline\\
$(0,4,4,4)$ & $f(\z,u),~~ |\alpha| - \frac54 |\tau| = 0$  \\ [1ex] 
  $(0,3,4,4)$ & $f(\z,u),~~ |\alpha| - \frac54 |\tau| \neq 0$  \\ [1ex] 
  $(0,2,4,4)-0$ & $- \frac{F(u)^2}{16}e^{\frac{4(\z-F_0(u))}{iF(u)}} + g(u) \z + g_0(u)$  \\ [1ex] 
  $(0,2,4,4)-1$ & $\frac{c^2}{16} e^{4\left(\frac{\z}{c}-\frac{iF_1(u)}{|c|^2}\right)} + g_1(u) \z + g_2(u) $  \\ [1ex] 
  $(0,2,4,4)-2$ & $f_2(\z-C_0 - i F_3(u))+g_3(u)z + g_4(u) $  \\ [1ex] 
  $(0,1,4,4)$ & $\frac{C_0^2}{16} e^{\frac{-4i(\z+C_1)}{C_0}} + f_1(u) \z + f_2(u) $  \\ [1ex] 
  $(0,1,3,4,4)-1.0$ & $\frac{C_0^2}{16} e^{\frac{-4i(\z+C_1)}{C_0}} +F_y[(C_2 + i) \z + 2C_3 + ln(F_y^{ \frac{C_0}{2}})] $,  \\ [1ex]
  & $F_y(u) \neq Cu^{-2}$  \\ [1ex]
  $(0,1,3,4,4)-1.1$ & $\frac{C_0^2}{16} e^{\frac{-4i(\z+C_1)}{C_0}} + F_x( \z+ C_2)$,  \\ [1ex]
  & $F_x(u) \neq Cu^{-2}$  \\ [1ex]
\hline \\
\end{tabular}
\caption{ All Vacuum Kundt waves admitting no symmetries}
\label{table:KundtWavesG0}
\end{center}
\begin{center} 
\begin{tabular}{ccc }
Invariant Count & $f(\z,u)$ &  Killing vector   \\ [0.5ex]
\hline\\
  $(0,2,3,3)-1$ & $\frac{c^2}{16} e^{\frac{4(\z-C_0-iC_1u)}{c}} + c_2\z + Im(c_2)C_1u+C_3$ & $U-C_1T$ \\ [1ex]
  $(0,2,3,3)-2$ & $f(\z-C - i C_0u)+c_1\z + Im(c_1)C_0u+C_2$ & $U-C_0T$  \\ [1ex]
  $(0,1,3,3)$ & $\frac{C_0^2}{16} e^{\frac{-4i(\z-iC_0u+C_1)}{C_0}} +c_3\z + Im(c_3)C_0u+C_2$ & $U-C_0T$  \\ [1ex]
\hline \\
 $(0,1,2,2)$ & $\frac{C_0^2}{16} e^{\frac{-4i(\z+C_1)}{C_0}}$ & $U$ and  \\ [1ex] 
  & & $T+ C_0^{-1}R$ \\ [1ex] 
\hline \\
\end{tabular}
\caption{ All Vacuum Kundt waves admitting symmetries; the Killing vectors are: $ U = \frac{\partial}{\partial u}$,$R = i \left(\frac{\partial}{\partial \z} - \frac{\partial}{\partial \bz}\right)$ and $T=\frac{u}{2}\frac{\partial}{\partial u}.$}   
\label{table:KundtWavesG1}
\end{center}
\end{table}
It has been shown that any spacetime that is not (locally) homogeneous requires at most $q \leq 7$. In fact, in the cases of Petrov types I, II and III it is known that {\bf at most} $q\leq 5$. The remaining Petrov types D, N and O provide instances where the upper bound may be higher. The type D vacuum spaces have been studied exhaustively \cite{Aman,CIV}  and shown to have an upper bound $q\leq 3$; type D non-vacuum have been shown to have $q\leq 6$ \cite{CI}. Similarly the type O spaces been analyzed extensively and in these spaces $q\leq 5$ \cite{Bradly,Siexas, Koutras, Invskea}

We are left with the Petrov type N spaces. As mentioned previously, the upper bound for type N vacuum spaces was $q\leq 5$ \cite{CIV, MRVickers}. Following from this work on vacuum type N spacetimes the only candidates for a vacuum type N space with $4\leq q \leq 5$ would be the Kundt vacuum waves; we have shown that the vacuum Kundt waves have $q\leq 4$. The addition of matter complicates the analysis; it has been shown that there is a non-vacuum Kundt wave with $q\leq 5$ \cite{skea} while the addition of $\Lambda$ can raise the upper bound up to seven \cite{KarlSharp}.

In \cite{Phd} a partial invariant classification was made of the type N plane-fronted waves (that is, all type N spacetimes admitting a non-twisting, shear-free null geodesic vector $\ell$ with cosmological constant and admitting  pure radiation, null Maxwell-Einstein or vacuum as sources). These spaces are interesting as they belong to the $CSI_\Lambda$ class and cannot be classified using polynomial scalar curvature invariants. Furthermore, these spaces are easily interpreted physically using the equations of geodesic deviation  due to the simple form that the curvature tensor takes in these spaces. The vacuum Kundt waves have been studied using representative timelike geodesics to study the structure of these spaces and the singularities in them \cite{PodolskyBelan}. The relationship between the geodesic deviation equations (i.e., the physical interpretation) and the Cartan invariants of a space is not known; however, for the type N $CSI_\Lambda$ spaces this is achievable, as illustrated by the analysis of the geodesic deviation equations in the vacuum plane wave spacetimes \cite{CMcMi}.  

In the future, we will examine the invariant classification of the $CSI_{\Lambda}$ spacetimes using the Karlhede algorithm. Alternatively, we could extend this approach to classify the pp-waves and Kundt waves in higher dimensions \cite{CCNV, HDpp}.
\end{section}
\begin{section}*{Acknowledgments} 
The authors would like to thank Jiri Podolsky for helpful comments, and Brendan Rutherford for his effort and keen eye in the editing process. This work was supported by NSERC of Canada. 
\end{section}

\begin{section}*{Appendix A: Vacuum Kundt Waves Admitting No Symmetry} 

In this appendix we collect all of the necessary invariants required to sub-classify the vacuum Kundt waves admitting no Killing vectors, by identifying the functionally independent invariants and those functionally dependent invariants that are not generic to all vacuum Kundt waves in this subclass, denoted using the arbitrary real-valued and complex-valued functions $\tilde{Z}$ and $\tilde{z}$ respectively. These functions constitute the essential classifying manifold, as all other curvature components to any order may be expressed in terms of these functions and their derivatives. In each list the use of a semi-colon indicates the separation between the set of invariants arising at each iteration of the Karlhede algorithm. 

\begin{prop} \label{prop:0444} 
The metrics belonging to the $(0,4,4,4)$ class may contain any analytic function, $f(z,u)$, not listed in the class of vacuum Kundt waves with invariant-counts beginning with $(0,n,...)$, $n<3$. 

Using the special coordinates $a=\frac14 ln( f_{,\z \z})$, the four functionally independent invariants may be constructed from the spin-coefficients in \eqref{SpinCoef1}: 
\beq & a-\ba,~~\z_{,a},~~|\tau|^{-1},~~v. & \nonumber \eeq
\noindent The classifying functions at first and second order are:
\beq & a+\ba= \tilde{Z}_0(i(a-\ba), \z_{,a},|\tau|^{-1}),~~\z_{,u}= \tilde{z}_1(i(a-\ba), \z_{,a},|\tau|^{-1}), &\nonumber \\ 
& \frac{24|\z_{,a}|^2}{16} + |\z_{,a}||\tau|^{-1}(\z_{,a}+\bz_{,\ba})+ |\tau|^{-2} ;&  \nonumber \\
&\z_{,aa} = \tilde{z}_2(i(a-\ba), \z_{,a},|\tau|^{-1}),~~\tilde{f}(a,u) = \tilde{z}_3(i(a-\ba), \z_{,a},|\tau|^{-1}).& \nonumber \eeq

\noindent The invariant coframe arises from the coframe \eqref{Kacoframe} by using the null rotation with parameters $B'$ and $B''$ to satisfy the conditions at first and second order respectively:  \beq Im\left(\gamma+B'\alpha + \frac54 \bar{B}'\tau\right) = 0;~~\Delta''(a-\ba) = 0. \nonumber \eeq
\end{prop}
\begin{prop} \label{prop:0344} 
The metrics belonging to the $(0,3,4,4)$ class may contain any analytic function, $f(z,u)$, not listed in the class of vacuum Kundt waves with invariant-counts beginning with $(0,n,...)$, $n<3$. 

Using the special coordinates $a=\frac14 ln( f_{,\z \z})$, the four functionally independent invariants may be constructed from the spin-coefficients in \eqref{SpinCoef1} even though the last invariant appears at second order: 
\beq & \z_{,a},~~\bz_{,\ba},~~|\tau|^{-1};~~v. & \nonumber \eeq
\noindent The classifying functions at first and second order are:
\beq & a= \tilde{z}_0(\z_{,a}, \bz_{,\ba},|\tau|^{-1}); &\nonumber \\ 
&\z_{,u} = \tilde{z}_2(\z_{,a}, \bz_{,\ba},|\tau|^{-1}),~~\z_{,aa}= \tilde{z}_3(\z_{,a}, \bz_{,\ba},|\tau|^{-1}),~~\tilde{f}(a,u) = \tilde{z}_4(\z_{,a}, \bz_{,\ba},|\tau|^{-1}).& \nonumber \eeq

\noindent The invariant coframe is found at first order by applying a null rotation to the coframe \eqref{Kaframe} with parameter $B'$ satisfying the conditions: $ \gamma+B'\alpha + \frac54 \bar{B}'\tau = 0$ which is explicitly given in proposition \ref{prop:KW2ndOInv}.
\end{prop}
\begin{prop} \label{prop:0244A} 
The metric belonging to the $(0,2,4,4)-0$ class has the canonical form for $f(\z,u)$
\beq & - \frac{F(u)^2}{16}e^{\frac{4(\z-F_0(u))}{iF(u)}} + g(u) \z + g_0(u), & \nonumber \eeq
\noindent where $F$, $F_0$, $g$ and $g_0$ are arbitrary functions of $u$. 

Using the special coordinates $a=\frac14 ln( f_{,\z \z})$, the four functionally independent invariants may be constructed from the spin-coefficients in \eqref{SpinCoef1} even though the last two invariants arise at second order:
\beq & a-\ba,~~\z_{,a} = iF(u)~with~F_{,u} \neq 0;~~ a+\ba,~~v. & \nonumber \eeq
\noindent The classifying functions at first and second order are:
\beq & |\tau|^{-1} = i\z_{,a}(a-\ba)+F_0(u); &\nonumber \\ 
&\z_{,aa} = 0,~~F_{,u}(u),~~g(u),~~\bar{g}(u),~~Re(g_0)(u).& \nonumber \eeq

\noindent The invariant coframe is found at first order by applying a null rotation to the coframe \eqref{Kaframe} with parameter $B'$ satisfying the conditions: $ \gamma+B'\alpha + \frac54 \bar{B}'\tau = 0$ which is explicitly given in proposition \ref{prop:KW2ndOInv}.
\end{prop}
\begin{prop} \label{prop:0244B} 
The metric belonging to the $(0,2,4,4)-1$ class has the canonical form for $f(\z,u)$
\beq & f(\z,u) = \frac{c^2}{16} e^{\frac{4(\z}{c}-\frac{iF_1(u)}{|c|^2}} + g_1(u) \z + g_2(u),~~Re(c) \neq 0 & \nonumber \eeq
\noindent where $F_1$, $g_1$, and $g_2$ are arbitrary functions of $u$. 

Using the special coordinates $a=\frac14 ln( f_{,\z \z})$, the four functionally independent invariants may be constructed from the spin-coefficients in \eqref{SpinCoef1} even though the last two invariants arise at second order:
\beq & a-\ba,~~|\tau|^{-1};~~ M_0,~~M_2 & \nonumber \eeq
\noindent where $M_0$ and $M_2$ are defined in proposition \ref{prop:KW2ndOInv}. The classifying functions at first,  second and third order are:
\beq & \z_{,a} = c; &\nonumber \\ 
&\z_{,aa} = 0,~~a+\ba=\tilde{Z}_0(a-\ba,|\tau|^{-1}, M_0,~~M_2),~~v=\tilde{Z}_1(a-\ba,|\tau|^{-1}, M_0, M_2);& \nonumber \\
& \Delta M_0 = i \tilde{Z}_2(a-\ba,|\tau|^{-1}, M_0, M_2),~~\Delta M_2 = \tilde{Z}_4(a-\ba,|\tau|^{-1}, M_0, M_2).& \nonumber \eeq

\noindent The invariant coframe is found at first order by applying a null rotation to the coframe \eqref{Kaframe} with parameter $B'$ satisfying the conditions: $ \gamma+B'\alpha + \frac54 \bar{B}'\tau = 0$ which is explicitly given in proposition \ref{prop:KW2ndOInv}.
\end{prop}


\begin{prop} \label{prop:0244C} 
The metric belonging to the $(0,2,4,4)-2$ class has the canonical form for $f(\z,u)$
\beq & f(\z,u) = f_2(\z-c_0 - i F_3(u))+g_3(u)\z + g_4(u) & \nonumber \eeq
\noindent where $F_3$, $g_3$, and $g_4$ are arbitrary functions of $u$. 

Using the special coordinates $a=\frac14 ln( f_{,\z \z})$, the four functionally independent invariants may be constructed from the spin-coefficients in \eqref{SpinCoef1} even though the last two invariants arise at second order:
\beq & a-\ba,~~|\tau|^{-1};~~ M'_1,~~\bar{M}'_1 & \nonumber \eeq
\noindent where $Z'_1 = M_1M_3^{-1}$ as defined in proposition \ref{prop:KW2ndOInv}. The classifying functions at first, second and third order are:
\beq & \z_{,a}=i\tilde{z}_0(a-\ba, |\tau|^{-1},M_1', \bar{M}_1'); &\nonumber \\ 
&a+\ba = \tilde{Z}_1(a-\ba, |\tau|^{-1}),~~v=\tilde{Z}_2(a-\ba, |\tau|^{-1}, Z'_1, \bar{Z}'_1);& \nonumber \\
& \Delta M'_1 = i \tilde{z}_3(a-\ba, |\tau|^{-1}, Z'_1, \bar{Z}'_1).& \nonumber \eeq

\noindent The invariant coframe is found at first order by applying a null rotation to the coframe \eqref{Kaframe} with parameter $B'$ satisfying the conditions: $ \gamma+B'\alpha + \frac54 \bar{B}'\tau = 0$ which is explicitly given in proposition \ref{prop:KW2ndOInv}.
\end{prop}


\begin{prop} \label{prop:0144} 
The metric belonging to the $(0,1,4,4)$ class has the canonical form for $f(\z,u)$
\beq & f(\z,u) = \frac{C_0^2}{16} e^{\frac{-4i(\z+C_1)}{C_0}} + f_1(u) \z + f_2(u). & \nonumber \eeq
\noindent where $f_1$ and $f_2$ may be any set of functions except those listed in the remaining invariant classes $(0,1,3,4,4)$, $(0,1,3,3)$ and $(0,1,2,2)$.  

Using the special coordinates $a=\frac14 ln( f_{,\z \z})$, the four functionally independent invariants may be constructed from the spin-coefficients in \eqref{SpinCoef1} even though the last three invariants arise at second order:
\beq & a-\ba;~~v,~~N_0,~~N_1 & \nonumber \eeq
\noindent where $N_0$ and $N_1$ are defined in equation \eqref{AclassNSharpNoG}. The classifying functions at first, and second order are:
\beq & \z_{,a}=iC_0,~~|\tau|^{-1} = iC_0(a-\ba)+2C_1; &\nonumber \\ 
&a+\ba = \tilde{Z}_0(a-\ba, N_0, N_1);& \nonumber \\
& \Delta N_0 = \tilde{Z}_1(a-\ba, v, N_0, N_1),~~\Delta N_1 = \tilde{Z}_2(a-\ba, v, N_0, N_1).& \nonumber \eeq

\noindent The invariant coframe is found at first order by applying a null rotation to the coframe \eqref{Kaframe} with parameter $B'$ satisfying the conditions: $ \gamma+B'\alpha + \frac54 \bar{B}'\tau = 0$ which is explicitly given in proposition \ref{prop:KW2ndOInv}.
\end{prop}


\begin{prop} \label{prop:01344A}
The metric belonging to the $(0,1,3,4,4)-1.0$ class has the canonical form for $f(\z,u)$
\beq & f(\z,u) = \frac{C_0^2}{16} e^{\frac{-4i(\z+C_1)}{C_0}} + F_y[(C_2 + i) \z + 2C_3 + ln(F_y^{ \frac{C_0}{2}})] & \nonumber \eeq
\noindent where $F_y$ may be any function except $Cu^{-2}$.  

Using the special coordinates $a=\frac14 ln( f_{,\z \z})$, the four functionally independent invariants may be constructed from the spin-coefficients in \eqref{SpinCoef1}:
\beq & a-\ba;~~v,~~N_2;~~F(u) = \frac{F_{y,u}}{F_y\frac32}	 & \nonumber \eeq

\noindent where $N_2$ is defined in equation \eqref{q4A}. The classifying functions at first, and second order are:
\beq & \z_{,a}=iC_0,~~|\tau|^{-1} = iC_0(a-\ba)+2C_1; &\nonumber \\ 
&N_1 = \left[ \frac{C_2}{|\tau|} + \frac{C_0^2|\tau|^2-2}{C_0|\tau|^2} \right] N_2;& \nonumber \\
&F_y = \tilde{Z}_0(F),~~a+\ba = \frac12 ln\left( \frac{N_2}{F_y} \right).& \nonumber \eeq

\noindent The invariant coframe is found at first order by applying a null rotation to the coframe \eqref{Kaframe} with parameter $B'$ satisfying the conditions: $ \gamma+B'\alpha + \frac54 \bar{B}'\tau = 0$ which is explicitly given in proposition \ref{prop:KW2ndOInv}.
\end{prop}


\begin{prop} \label{prop:01344B} 
The metric belonging to the $(0,1,3,4,4)-1.1$ class has the canonical form for $f(\z,u)$
\beq & f(\z,u) = \frac{C_0^2}{16} e^{\frac{-4i(\z+C_1)}{C_0}} +  F_x( \z+ C_2) & \nonumber \eeq
\noindent where $F_x$ may be any function except $Cu^{-2}$.  

Using the special coordinates $a=\frac14 ln( f_{,\z \z})$, the four functionally independent invariants may be constructed from the spin-coefficients in \eqref{SpinCoef1}:
\beq & a-\ba;~~v,~~N_2;~~F(u) = \frac{F_{x,u}}{F_x\frac32}	 & \nonumber \eeq 

\noindent where $N_2$ is defined in equation \eqref{q4A}. The classifying functions at first, and second order are:
\beq & \z_{,a}=iC_0,~~|\tau|^{-1} = iC_0(a-\ba)+2C_1; &\nonumber \\ 
&N_1 = \frac{-N_2}{C_2|\tau|};& \nonumber \\
&F_x = \tilde{Z}_0(F).& \nonumber \eeq

\noindent The invariant coframe is found at first order by applying a null rotation to the coframe \eqref{Kaframe} with parameter $B'$ satisfying the conditions: $ \gamma+B'\alpha + \frac54 \bar{B}'\tau = 0$ which is explicitly given in proposition \ref{prop:KW2ndOInv}.
\end{prop}


\end{section}

\begin{section}*{Appendix B: Vacuum Kundt Waves Admitting a Symmetry}

In this appendix we collect all of necessary invariants required to sub-classify the vacuum Kundt waves admitting one Killing vectors, by identifying the functionally independent invariants and those functionally dependent invariants that are not generic to all vacuum Kundt waves in this subclass, denoted using the arbitrary real-valued and complex-valued functions $\tilde{Z}$ and $\tilde{z}$ respectively. These functions constitute the essential classifying manifold, as all other curvature components to any order may be expressed in terms of these functions and their derivatives. In each list the use of a semi-colon indicates the separation between the set of invariants arising at each iteration of the Karlhede algorithm. 

\begin{prop} \label{prop:0233B}  
The metric belonging to the $(0,2,3,3)-1$ class has the canonical form for $f(\z,u)$
\beq & \frac{c^2}{16} e^{\frac{4(\z-C_0-iC_1u)}{c}} + c_2\z + Im(c_2)C_1u+C_3 & \nonumber \eeq
\noindent where $c, c_2$ and $C_0,C_1,C_3$ are arbitrary complex-valued and real-valued functions respectively.

Using the special coordinates $a=\frac14 ln( f_{,\z \z})$, the four functionally independent invariants may be constructed from the spin-coefficients in \eqref{SpinCoef1}:
\beq & a-\ba,~~a+\ba;~~v & \nonumber \eeq
\noindent where $M_0$ and $M_2$ are defined in proposition \ref{prop:KW2ndOInv}. The classifying functions at first,  second and third order are:
\beq & \z_{,a} = c,~~|\tau|^{-1} = Re(c)(a+\ba)+Im(c)(a-\ba)+C_0; &\nonumber \\ 
& C_1,~~c_2,~~C_3.  & \nonumber \eeq

\noindent The invariant coframe is found at first order by applying a null rotation to the coframe \eqref{Kaframe} with parameter $B'$ satisfying the conditions: $ \gamma+B'\alpha + \frac54 \bar{B}'\tau = 0$ which is explicitly given in proposition \ref{prop:KW2ndOInv}.
\end{prop}


\begin{prop} \label{prop:0233C} 
The metric belonging to the $(0,2,3,3)-2$ class has the canonical form for $f(\z,u)$
\beq & f(\z-C - i C_0u)+c_1\z + Im(c_1)C_0u+C_2 & \nonumber \eeq
\noindent where $C,C_0,C_2$,and $c_1$ are arbitrary real-valued and complex-valued constants.
Using the special coordinates $a=\frac14 ln( f_{,\z \z})$, the four functionally independent invariants may be constructed from the spin-coefficients in \eqref{SpinCoef1}:
\beq & \z_{,a},~~\bz_{,\ba};~~v. & \nonumber \eeq
\noindent The classifying functions at first, second and third order are:
\beq & a-\ba=i\tilde{Z}_0(\z_{,a},\bz_{,\ba}),~~\z+\bz = \tilde{Z}_1(\z_{,a},\bz_{,\ba}),~~C; &\nonumber \\ 
&a+\ba = \tilde{Z}_2(\z_{,a},\bz_{,\ba}),~~c_1,~~C_2.& \nonumber \eeq

\noindent The invariant coframe is found at first order by applying a null rotation to the coframe \eqref{Kaframe} with parameter $B'$ satisfying the conditions: $ \gamma+B'\alpha + \frac54 \bar{B}'\tau = 0$ which is explicitly given in proposition \ref{prop:KW2ndOInv}.
\end{prop}


\begin{prop} \label{prop:0133A}
The metric belonging to the $(0,1,3,3)$ class has the canonical form for $f(\z,u)$
\beq & \frac{C_0^2}{16} e^{\frac{-4i(\z-iC_0u+C_1)}{C_0}} +c_3\z + Im(c_3)C_2u+iC_2 & \nonumber \eeq
\noindent where $C_0,C_1,C_2,$ and $c_3$ are arbitrary real and complex valued constant, respectively.

Using the special coordinates $a=\frac14 ln( f_{,\z \z})$, the four functionally independent invariants may be constructed from the spin-coefficients in \eqref{SpinCoef1}:
\beq & a-\ba;~~v,~~u^{-2}e^{-2a-2\ba} & \nonumber \eeq 

\noindent The classifying functions at first, and second order are:
\beq & \z_{,a}=iC_0,~~|\tau|^{-1} = iC_0(a-\ba)+2C_1; &\nonumber \\ 
& C_2,~~c_3.& \nonumber \eeq

\noindent The invariant coframe is found at first order by applying a null rotation to the coframe \eqref{Kaframe} with parameter $B'$ satisfying the conditions: $ \gamma+B'\alpha + \frac54 \bar{B}'\tau = 0$ which is explicitly given in proposition \ref{prop:KW2ndOInv}.
\end{prop}


\end{section}

\begin{section}*{Appendix C: Vacuum Kundt Waves Admitting Two Symmetries}

In this appendix we collect all of necessary invariants required to sub-classify the vacuum Kundt waves admitting two Killing vectors, by identifying the functionally independent invariants and those functionally dependent invariants that are not generic to all vacuum Kundt waves in this subclass. These functions constitute the essential classifying manifold, as all other curvature components to any order may be expressed in terms of these functions and their derivatives. In each list the use of a semi-colon indicates the separation between the set of invariants arising at each iteration of the Karlhede algorithm.

\begin{prop} \label{prop:0122}
The metric belonging to the $(0,1,2,2)$ class has the canonical form for $f(\z,u)$
\beq & f(\z,u) = \frac{C_0^2}{16} e^{\frac{-4i(\z-iC_2+C_1)}{C_0}} & \nonumber \eeq
\noindent where $C_0$ and $C_1$ are arbitrary real-valued constants.

Using the special coordinates $a=\frac14 ln( f_{,\z \z})$, the four functionally independent invariants may be constructed from the spin-coefficients in \eqref{SpinCoef1}:
\beq & a-\ba;~~e^{-a-\ba}v.& \nonumber \eeq 

\noindent The classifying functions at first and second order are:
\beq & \z_{,a}=iC_0,~~|\tau|^{-1} = iC_0(a-\ba)+2C_1. &\nonumber \eeq

\noindent The invariant coframe is found at first order by applying a null rotation to the coframe \eqref{Kaframe} with parameter $B'$ satisfying the conditions: $ \gamma+B'\alpha + \frac54 \bar{B}'\tau = 0$ which is explicitly given in proposition \ref{prop:KW2ndOInv}.
\end{prop}


\end{section}

\begin{section}*{Appendix D: All Potential Invariant Counts for the Vacuum Kundt Waves}
To write down a potential case of the Karlhede algorithm up to a given iteration, p, we will use the following notation, $(t_1,t_2,...,t_p,...)$, where $t_i,~ i\in[1,p]$ denotes the number of functionally independent invariants at the i-th iteration of the Karlhede algorithm. We may map out all potential cases of the Karlhede algorithm, by using each potential invariant count as a node in a tree diagram where dashed lines indicate the existence of a non-trivial isotropy group at the previous iteration, while  solid lines indicate all isotropy has been fixed. 
 \begin{figure}[h] 
  \centering 
  \includegraphics[scale=0.35]{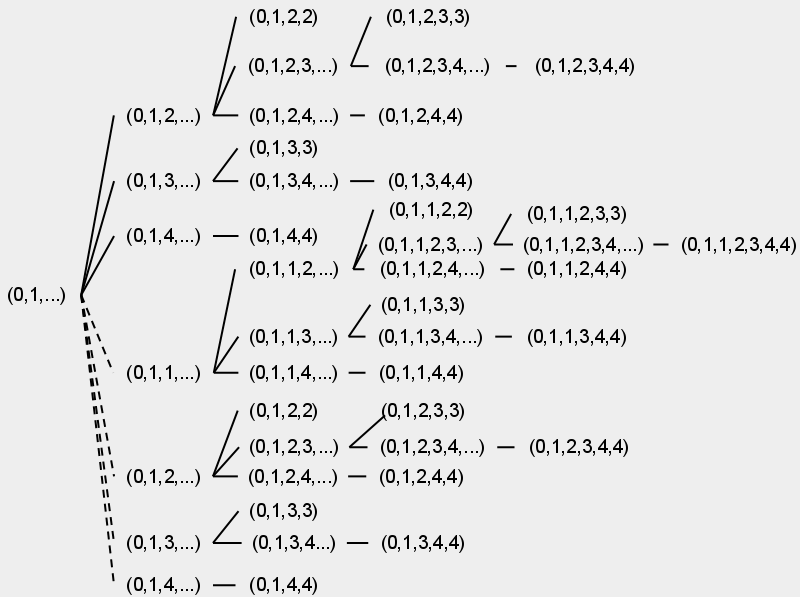} 
 \caption{Potential invariant-count tree for the case where one functionally independent invariant appears at first order of the algorithm } \label{0KarlAlg}
\end{figure}
 \begin{figure}[h] 
  \centering 
  \includegraphics[scale=0.35]{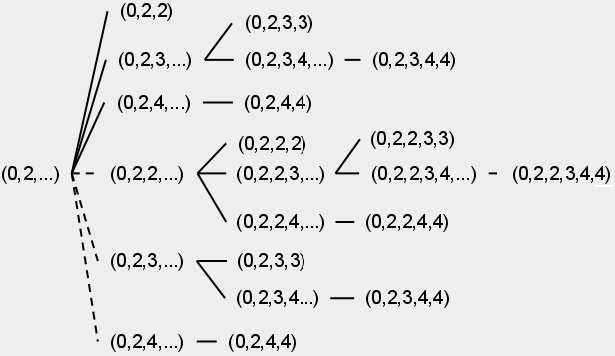} 
 \caption{Potential invariant-count tree for the case where two functionally independent invariants appear at first order of the algorithm } \label{2KarlAlg}
\end{figure}
  \begin{figure}[h] 
  \centering 
  \includegraphics[scale=0.35]{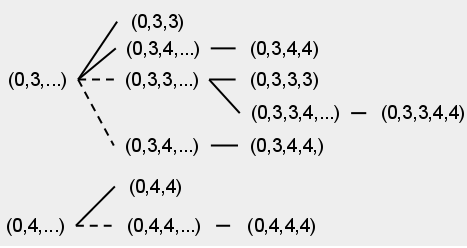} 
 \caption{ Potential invariant-count trees for the case where three or four functionally independent invariants appear at first order of the algorithm } \label{1KarlAlg}
\end{figure}
\newpage
\end{section}

\end{document}